\documentclass[12pt]{article}
\usepackage{amsmath}
\usepackage{amssymb}
\usepackage{amsbsy}
\usepackage{setspace} 
\usepackage{xr-hyper}
\usepackage{hyperref}
\makeatletter\renewcommand{\hyper@linkfile}[3]{#1}\makeatother 
\usepackage{pdflscape}
\usepackage{graphicx}
\usepackage{booktabs}
\usepackage{array}
\usepackage{multirow}
\usepackage{natbib}
\usepackage[toc,page]{appendix}
\usepackage{rotating}
\usepackage{enumerate}
\usepackage{enumitem}
\usepackage{float}
\usepackage{commath}
\usepackage{caption}
\usepackage{subcaption}
\usepackage{mathtools}
\usepackage[dvipsnames]{xcolor}
\hypersetup{colorlinks = true, citecolor =BrickRed,
  linkcolor= BrickRed, urlcolor = BrickRed}
\usepackage[sc]{mathpazo}
\usepackage[T1]{fontenc}
\definecolor{drkblue}{HTML}{000080}
\newtheorem{proposition}{Proposition}
\newtheorem{definition}{Definition}

\newtheorem{lemma}{Lemma}
\newenvironment{proof}[1][Proof]{\begin{trivlist}
\item[\hskip \labelsep {\bfseries #1}]}{\end{trivlist}}

\newcommand{\uv}{\underline v}
\newcommand{\ux}{\underline x}

\newcommand{\phihat}{\hat{\phi}}

\newcommand{\R}{\mathbb R} 
\newcommand{\E}{\mathbb E} 
\newcommand{\ceq}{\mathrel{\mathop:}=} 
\newcommand{\eqc}{=\mathrel{\mathop:}} 

\providecommand{\possessivecite}[1]{\citeauthor{#1}'s\nolinebreak[2]
(\citeyear{#1})}

\begin{document}

\title{Optimal Project Management \footnote{We would like to thank
    Matteo Escudé, Johannes Hörner, Aaron Kolb, and Can Urgun for
    helpful comments and discussions. We also thank seminar
    participants at City University of London, Collegio Carlo Alberto,
    Essex, JHU, London Business School, Miami Business School, MIT, US Naval
    Academy, Nottingham, Penn, Princeton, Southeast Theory Festival,
    Surrey, and Toulouse.}}  \author{Alessandro Bonatti\thanks{Sloan
    School of Management, MIT,
    \protect\protect\protect\protect\protect\href{mailto:bonatti@mit.edu}{\texttt{bonatti@mit.edu}}}
  \hspace{3mm} \hspace{2mm} Doruk Cetemen\thanks{LUISS, RHUL and EIEF,
    \protect\protect\protect\protect\protect\href{mailto:dcetemen@luiss.it}{\texttt{dcetemen@luiss.it}}}
  \hspace{3mm}\hspace{2mm} Juuso Toikka\thanks{The Wharton School,
    University of Pennsylvania,
    \protect\protect\protect\protect\protect\href{mailto:toikka@wharton.upenn.edu}{\texttt{toikka@wharton.upenn.edu}}}}

\maketitle
\begin{abstract}
  We study optimal dynamic contracts for long-term projects where a
  risk-averse agent exerts hidden effort to increase the drift of a
  Brownian process toward a completion threshold. Relative to the
  first best, the optimal contract slows progress, makes success less
  likely, induces earlier termination, and narrows project
  scope. Allowing for costly project resets generates further
  distortions: resets may be excessive or insufficient depending on
  the residual effort required afterwards. With endogenous risk
  taking, optimal incentives introduce penalties for failure and
  reduce rewards for effort above a critical progress level, resulting
  in discontinuous and nonmonotonic effort in progress. In the
  baseline model, the optimal contract can be implemented with linear
  spot contracts that require no long-term commitment.

\vspace{2mm}


\vspace{2mm}


\end{abstract}\newpage

\section{Introduction}

Many long-term projects---from scientific research and drug
development to complex engineering and software initiatives---progress
through uncertain paths and often require previously unanticipated
steps. Success typically depends on a sequence of decisions made under
limited observability and evolving conditions. Managing such projects
involves dynamic trade-offs between motivating effort, tolerating
setbacks, and making strategic choices about when to persevere, reset,
or take on risk.

In this paper, we develop a flexible dynamic contracting framework to
study optimal project management. We consider a principal who must
decide how to incentivize an agent to exert effort; how tolerant she
should be of failure before pulling the plug or shifting the approach;
how ambitious the project should be in the first place; and when to
allow the agent to take on additional risk in order to speed up
progress.

To capture the rich uncertainty inherent to complicated projects, we
model progress as a Brownian motion that must reach a threshold, or
target, for completion. The agent can increase its drift by exerting
hidden, costly effort. The distance to the target measures the
remaining effort required to complete the project, though the actual
effort needed may turn out to be smaller or larger depending on the
resolution of uncertainty.

We assume that the agent has constant absolute risk aversion (CARA)
preferences and can privately save and borrow on the same terms as the
principal. This allows us to abstract from wealth effects. Further, it
eliminates the principal's ability to provide incentives via
distorting the agent's intertemporal consumption smoothing---a feature
of optimal contracts in models without private savings that seems of
limited relevance for managing long-term projects. Moreover, optimal
contracts in our setting have several natural implementations, e.g.,
via deferred pay or short-term contracts.

The CARA preferences with hidden savings imply that instead of two
state variables, progress and continuation utility, the principal's
problem can be formulated recursively with progress as the only state
variable. In most variants we study, the principal's expected payoff
is characterized by the value of an auxiliary first-best problem where
effort is contractible but its cost to the agent is scaled up by a
constant. The effects of agency frictions can thus be analyzed through
comparative statics with respect to the cost parameter. Even when the
first best and second best differ qualitatively, such as in the
extension to risk taking where the agent's hidden action is
two-dimensional, the analysis remains remarkably tractable.

\paragraph{Main Results} In the baseline model, the principal must
only decide how to reward the agent and when to stop the project. The
optimal contract features increasing effort in progress and a constant
termination cutoff. That is, it is optimal to make the agent work
harder the closer the target, and abandon the project in favor of an
outside option if progress falls sufficiently far from it due to
adverse shocks. Relative to the first best (which shares these
qualitative features), the optimal contract induces less effort and
exhibits lower failure tolerance in the form of a higher termination
cutoff. Thus, moral hazard delays completion and makes it more likely
the project be abandoned.

The optimal contract has a natural implementation where a single
history-dependent payment is made when the project succeeds or is
abandoned. Alternatively, it can be implemented by a sequence of spot
contracts, each of which is the optimal linear contract of the static
model of \citet{holmstrom1987aggregation} for the current marginal
value of progress. Linearity is thus a result rather than an
assumption. Moreover, under this implementation the agent is
indifferent between continuing and leaving at every history, and the
principal gains nothing from revising the terms, so the relationship
requires no long-term commitment on either side.

The baseline model takes the target as given, but in many settings
choosing the target is part of managing the project. For instance,
goals have to be set for research and development, and firms have to
decide how good a new product should be before it is
launched. Accordingly, we consider the case of endogenous project
scope where the principal can choose the target for the project, with
more ambitious targets yielding larger benefits upon completion. The
findings from the baseline model carry over readily to this
case. Moreover, we show that the principal optimally pursues less
ambitious projects than in the first best, which points to a simple
but novel form of agency cost.

We then enrich the model to incorporate endogenous project resets to
capture settings where there are multiple alternative approaches to
tackling the project. We assume that the principal can always abandon
a failing approach and switch to a new ex ante identical
one. Formally, this means that progress can be reset to a fixed
restart point by paying a fixed cost. Motivating examples include
abandoned early-stage AI models that are retrained from scratch
following poor performance, software teams discarding early versions
of a system, or clinical trials restarted after unfavorable interim
results.\footnote{See \citet{brooks1975mythical} and the discussion in
  \citet{cusumano1991japan}.}

The optimal contract in the model with resets implements a threshold
policy: below a certain progress level, the principal pays the fixed
cost to reset; otherwise, the project continues. Relative to the first
best, the second-best policy may reset too aggressively when the
restart point is close to the target, and too conservatively when the
restart point is far. This result reflects an intuitive tension: using
resets substitutes for effort, which is more costly to implement under
moral hazard; but resets alone do not complete the project. The more a
reset requires exerting further effort to complete the project, the
more the value of resetting is depressed when incentives are costly to
provide.

Finally, we study endogenous risk taking. To do so, we add to the
model the possibility of a breakdown---a Poisson shock that kills the
project. In addition to effort, the agent chooses between a risky and
a safe mode unobservably to the principal. The risky mode increases
the project's drift relative to the safe mode, but comes at the cost
of a higher hazard rate of a breakdown. This structure reflects
dynamic moral hazard in settings like early-phase pharmaceutical
development, where firms may fast-track compounds into clinical trials
to compress timelines, accepting higher risk of regulatory
failure.\footnote{Failure rates are high in practice: using
  clinical-development transitions observed during 2011--2020,
  \citet{thomas2021clinical} estimate a likelihood of approval from
  Phase I of 7.9\% overall and 5.3\% in oncology.} Similarly,
technology startups may push unstable features in pursuit of growth.

The optimal contract allows the agent to take additional risk to
accelerate development when the project is far from completion. By
contrast, when success is close, the downside risk of a breakdown
outweighs the gains from faster progress. Under additional
assumptions, we show that the contract has a cutoff structure: once
progress exceeds a threshold, the contract switches from the risky to
the safe mode. In the risky mode, the agent is not penalized for
breakdowns. In the safe mode, the principal deters risk taking by
imposing penalties for breakdowns and reducing rewards for progress.

Interestingly, the interaction between hidden effort and risk choices
implies that effort is nonmonotonic in progress: it rises as progress
approaches the cutoff, drops discontinuously at the cutoff, and then
increases again until completion. By contrast, first-best effort is
always increasing in progress.

A technical contribution of our analysis is verifying global incentive
compatibility for a CARA agent with hidden savings subject to a
standard pathwise no-Ponzi condition. We use localization arguments
that allow us to dispense with the ad hoc integrability conditions
(that implicitly restrict deviations) imposed in the
literature. \citet*{bloedel2023persistent} show an analogous result
for the pure consumption-savings problem when income follows an
Ornstein-Uhlenbeck process; our proof can be seen as a generalization
to controlled income processes with a finite stopping time. The
approach could prove useful in other contracting environments with
unbounded states.

\paragraph{Related Literature}

This paper contributes to the study of dynamic contracting in
continuous time following \citet{sannikov2008continuous}. As noted
above, we abstract from wealth effects and assume the agent can
privately save and borrow on the same terms as the principal.
\cite*{fudenberg1990short} show (in discrete time) that under these
assumptions short-term contracts are without loss of optimality. Most
importantly for us, the principal's value function is then separable
in progress and continuation utility, rendering the state effectively
one-dimensional. These assumptions are also adopted by
\cite{he2011model} to study the relationship between pay-performance
sensitivity and firm size; our model can be viewed as an extension of
his to a nonstationary setting.\footnote{Other papers adopting a
  similar approach include \cite{holmstrom1987aggregation},
  \cite{williams2015solvable}, \cite*{he2017optimal},
  \cite{marinovic2019ceo}, and \cite*{cetemen2023renegotiation}.}
While this case is known to be tractable, the general case of hidden
borrowing and savings is challenging---see, e.g.,
\citet{ditella2021optimal} and the references therein.

In terms of substantive questions, the most closely related papers are
those considering contracting for project
completion. \citet{mason2015getting} analyze a dynamic moral hazard
problem in which an agent exerts costly effort to increase the arrival
rate of a Poisson event that corresponds to project completion.
\citet{green2016breakthroughs} extend the model to a setting where the
friction is cash-flow diversion rather than hidden effort and where
two Poisson arrivals are required for completion, with the agent
privately observing the first one. \citet*{feng2024setbacks} use the
cash-flow diversion approach to study a project with deterministic
progress subject to random, privately observed Poisson setbacks that
erase all accumulated progress. In turn, \cite*{feng2025setbacks}
consider partial setbacks in a two-stage project. Our analysis
complements these papers by introducing endogenous project scope,
resets, and risk taking, and by allowing for a rich stochastic model
of progress similar to that in \possessivecite{georgiadis2015projects}
team-production game. (\cite{georgiadis2015projects} considers
contracting with limited instruments as an extension, but does not
solve for the fully optimal contract.) For other works considering
incentives for project completion, see, e.g.,
\citet{shavell1979optimal}, \citet{hopenhayn1997optimal},
\citet{kocherlakota2004figuring}, \citet{toxvaerd2006time}, and
\citet*{georgiadis2014project}.

Risk taking has been considered in a stationary dynamic contracting
context by \citet*{demarzo2013risking} who assume that, in addition to
diverting funds, the agent can boost short-term profits by exposing
the firm to the risk of a terminal breakdown. The authors show that
when the agent's continuation value falls sufficiently low due to poor
performance, he takes excessive risks. \citet{wong2019dynamic}
analyzes the case in which taking risks increases the chance of a
large loss rather than
termination.\footnote{\citet*{bromberg2021scale} introduce scale into
  the setup of \citet{wong2019dynamic}. See also
  \citet{rochet2016risky} and \citet{li2025optimal}.} Finally,
\cite*{biais2010large} and \cite{myerson2015moral} consider models in
which the agent exerts hidden effort to prevent the arrival of
disasters. In contrast to our model, the moral hazard problem in these
two papers is one-dimensional and there is no interaction between
exerting effort and taking risks.

Our analysis is related to contracting in Brownian settings where
nonstationarity arises because of learning, such as in
\cite{madsen2022designing} or \cite{demarzo2011learning}. In these
models the agent has private information about the profitability of
the venture either on path or following a deviation, whereas here
progress is publicly observable.

Also related is the literature on contracting for experimentation
starting from \citet{bergemann1998venture,
  bergemann2005financing}---see, e.g., \citet{manso2011motivating},
\citet{horner2013incentives}, \citet{guo2016dynamic},
\citet*{halac2016optimal}, and
\citet{mcclellan2022experimentation}. The key difference relative to
these papers is that while in our model the amount of effort needed to
complete the project is uncertain and only learned over time, it is
nevertheless always feasible (albeit not optimal) to reach the target
almost surely in finite time.

\paragraph{Outline} The rest of the paper is organized as
follows. Section~\ref{Sec:Model} presents the baseline model;
Section~\ref{Sec:Planner} solves the first-best benchmark;
Section~\ref{Sec:Analysis} derives the optimal contract;
Section~\ref{Sec:ResModel} examines resetting;
Section~\ref{Sec:DisModel} analyzes endogenous risk taking; and
Section~\ref{Sec:Conc} concludes. All proofs are in the Appendix.

\section{Baseline Model} \label{Sec:Model}

A risk-neutral principal (she) hires a risk-averse agent (he) to
complete a project. The project's progress is a publicly observable
stochastic process, which the agent controls with costly, hidden
effort. The events unfold in continuous time, indexed by
$t \in [0,\infty)$.

\paragraph{Project Technology} The project's progress $X$ is a
real-valued, controlled process, whose starting value is normalized to
zero (i.e., $X_0= 0$) and which evolves according to
\begin{equation} \label{progress}
dX_t = A_t\, dt + \sigma\, dZ_t,
\end{equation}
where $A_t \in \R_{+}$ is the agent's hidden effort at time $t$;
$\sigma >0$ is volatility; and $dZ_t$ is the increment of a standard
Brownian motion $Z$, unobservable to both parties.

The project is completed when progress reaches the target
$\bar x>0$, i.e., the first time $X_t = \bar x$. The project can be
terminated at any time before completion.

Beyond its literal interpretation, the project’s progress can be
viewed as the gradual resolution of uncertainty about the remaining
effort required for completion, starting from an expected requirement
of $\bar x-X_t$ at time $t$. (The target $\bar x$ is then the time-$0$
expected requirement.) In this interpretation, upward and downward
movements in progress reflect the arrival of information about the
realized residual work associated with completing intermediate
tasks. The Brownian shocks capture incremental learning about this
residual requirement as individual steps are undertaken.\footnote{For
  a related model of continuous-time gradual learning, see
  \cite*{brsvb12}.}

\paragraph{Payoffs} The risk-neutral principal receives a lump-sum
benefit $b>0$ upon project completion. If instead the project is
terminated before it is completed, she receives a salvage value
$\underline v \in (0, b)$. Thus, if the project is stopped (i.e.,
completed or terminated) at a random time $\tau$, her payoff is given
by
\begin{equation*}
  e^{-r \tau} \left(\mathbf{1}_{\{X_\tau\geq\bar{x}\}}b
    +\mathbf{1}_{\left\{X_{\tau}<\bar{x}\right\}}\underline v\right)
  -  \int_0^{\infty} e^{-r t } C_t dt,
\end{equation*}
where $r >0$ is the interest rate and $C_t$ is the time-$t$ payment to
the agent, which may be positive or negative. (Payments can continue
even after the project ends.)

The agent is risk-averse with CARA flow utility. His flow cost of
effort is measured in monetary terms by a cost function
$h:\R_+ \to \R_+$ with $h(a) = \phi a^2/2$ for some $\phi > 0$,
assumed quadratic for simplicity.\footnote{\label{fn:cost}Our results
  extend qualitatively to any effort cost $h$ with
  $h(0) = h_a(0) = 0$, $h_{aa} > 0$, and $h_{aaa} \geq 0$, except that
  in Propositions~\ref{Prop:FirstBest}\ref{it:d05} and
  \ref{Prop:SecondBest}\ref{it:s04} it is then the marginal value of
  progress, rather than effort, that grows in expectation at rate $r$,
  and the sufficient condition in
  Proposition~\ref{Prop:OptimalContract} takes a different form.} The
agent can save and borrow, and thus his time-$t$ consumption,
$\hat C_t$, may in general be different from the payment $C_t$
received from the principal. Denoting his time-$t$ effort by
$\hat A_t$, the agent's payoff is given by
 \begin{equation*}
  \int_0^{\infty} e^{-r t }u(\hat A_t,\hat C_t) dt,
\end{equation*}
where
\begin{equation}
  u(\hat a,\hat c) \ceq - \frac{1}{\eta}
  \exp\{- \eta (\hat c-h(\hat a))\}\label{carau}
  = - \frac{1}{\eta} \exp\left\{- \eta
    \left(\hat c-\frac{\phi}{2}\hat a^2\right)\right\}.
\end{equation}

Both parties can save and borrow at the rate $r$. The agent's
consumption and savings decisions are unobservable to the principal.

\paragraph{Contracting} The parties commit at time $0$ to a long-term
contract that specifies monetary compensation for the agent as a
function of the project's progress. The contract also specifies the
conditions (if any) under which the project be terminated.

Formally, a contract is a triple $(A,C,\tau)$ consisting of a
recommended effort process $A$ taking values in $\R_+$, a payment
process $C$ taking values in $\R$, and a stopping time $\tau$ taking
values in $\R_+$, all progressively measurable with respect to
progress $X$. The stopping time $\tau$ captures both project
completion and early termination. It is thus required to satisfy
$\tau \leq \inf\{t \geq 0: X_t = \bar x\}$ almost surely.

Given a contract $(A,C,\tau)$, the agent chooses an effort process
$\hat A$ and a consumption process $\hat C$, both progressively
measurable with respect to progress, to maximize his utility:
\begin{equation}\label{prob}
    U(A,C,\tau)\ceq \max_{\hat A , \hat C} \E^{\hat A} \left[ \int_0^\infty  e^{-r t} u(\hat A_t,\hat C_t) \, dt \right],
\end{equation}
where the expectation is conditional on the effort process
$\hat{A}$. The agent's choice of actions and consumption is subject to
the laws of motion for progress and savings, given by
\begin{eqnarray*}
	dX_t &=& \hat A_t\, dt + \sigma\, dZ_t, \quad X_0 = 0, \\
	dS_t & = & (r S_t + C_t - \hat C_t) \, dt, \quad S_0 = 0,
\end{eqnarray*}
and to the transversality condition
$\lim_{t\to \infty } e^{-r t} S_t = 0$ (a.s.) ruling out Ponzi
schemes. Savings increase through interest income $rS_t$ and payments
$C_t$, and decrease through consumption $\hat C_t$; initial savings
are normalized to zero without loss of generality.

It is convenient to measure the agent's lifetime utility levels in
monetary terms. For $U < 0$, let $CE(U)$ denote the wealth whose
annuity gives lifetime utility $U$. That is,
\begin{equation}\label{eq:CE}
  CE(U) \ceq -\frac{1}{\eta r} \ln(-\eta r U),
\end{equation}
so that $\int_0^\infty e^{-rt} u(0, r\, CE(U))\, dt = U$. We normalize
the outside option to zero wealth without loss of generality (cf.\
\citealp{holmstrom1987aggregation}; see the Supplemental Appendix for
details). Thus, the contract $(A,C,\tau)$ satisfies the participation
constraint if
\begin{equation}\label{ir}
  CE(U(A,C,\tau)) \geq 0.
\end{equation}
We impose the participation constraint only at time $0$ to simplify
the analysis. However, we show in Section~\ref{Sec:Implementation}
that the optimal contract can be taken to satisfy it at all times. (At
time $t$ the outside option is worth $S_t$ to an agent with savings
$S_t$; see Lemma~\ref{lem:AWE}.)

It is convenient to focus on contracts $(A,C,\tau)$ under which it is
optimal for the agent to obey the effort recommendation $A$ and
neither borrow nor save, i.e., to consume exactly $C$. In such a
contract, the payment process $C$ also acts as an incentive compatible
consumption recommendation. Formally, we define these contracts as
follows.

\begin{definition}
  A contract $(A,C,\tau)$ is \textbf{incentive compatible} if
  $(\hat{A},\hat C) = (A,C)$ is a solution to the agent's problem
  \eqref{prob}.
\end{definition}

The following lemma is standard in the literature on dynamic
contracts.\footnote{The idea is similar to the revelation
  principle. For a proof, see, e.g., \cite{he2011model}.}

\begin{lemma}
  For any contract with finite payoffs to the agent and the principal,
  there is an incentive-compatible contract that delivers the same
  payoffs to both parties as the original contract.
\end{lemma}

That is, focusing on incentive-compatible contracts is without loss of
generality as far as payoffs are concerned. Armed with this result,
the principal's problem of designing an optimal contract can be
formulated as
\begin{equation}\label{second-best}
  \max_{A, C, \tau} \mathbb{E}^{A} \bigg[ e^{-r \tau}
  \big(\mathbf{1}_{\{X_\tau\geq\bar{x}\}}b+\mathbf{1}_{\left\{X_{\tau}<\bar{x}\right\}} 
  \underline v \big) -  \int_0^{\infty} e^{-r t } C_t \, dt  \bigg],
\end{equation}
where the maximum is over incentive-compatible contracts that
satisfy participation \eqref{ir}.

It is worth noting that while the solution to \eqref{second-best} pins
down the agent's consumption, the payment policy is unique only up to
the net present value as the agent can save and borrow. In
Section~\ref{Sec:Implementation} we discuss implementations via
deferred pay or short-term contracts, the latter of which requires no
long-term commitment from either party.

\section{Observable Effort}\label{Sec:Planner}

Before deriving the optimal contract, it is useful to consider the
first-best benchmark where the agent's effort is observable. (Whether
consumption is also observable is immaterial.) In this case, the least
expensive way to satisfy the agent's participation constraint
\eqref{ir} is by giving him constant consumption net of effort
costs. This can be achieved by paying a flow wage that exactly
compensates for the required effort, i.e.,
\begin{equation*}
  C_t = \frac{\phi}{2}A_t^2  \quad
  \text{for all $t \in [0,\infty)$}.  
\end{equation*}

Zero effort is of course optimal after the project is stopped. Thus,
at any time $t$ and at any current progress level $x \leq \bar x$, the
first-best effort policy and stopping time solve
\begin{equation}\label{first-best}
  v(x) \ceq \max_{A, \tau} \mathbb{E}_{t,x}^{A}\bigg[ e^{-r (\tau-t)}
  \big(\mathbf{1}_{\{X_\tau \geq \bar{x}\}} b
  + \mathbf{1}_{\{X_{\tau} < \bar{x}\}} \underline v \big)
  - \frac{\phi}{2}\int_{t}^{\tau} e^{-r (s-t)} A_s^2\, ds \bigg].
\end{equation}
This first-best optimal control and stopping problem is a one-player
version of the team-production game in
\citet{georgiadis2015projects}.\footnote{See Remark 4 in
  \citet{georgiadis2015projects} for the case of a positive salvage
  value as in our model.} As such, part of the solution is a
termination threshold $\underline x^{FB} < \bar x$ at which the
project is optimally abandoned. In the continuation region
$[\underline x^{FB}, \bar x]$, the value function $v$ defined by
\eqref{first-best} is infinitely differentiable and satisfies the
Hamilton-Jacobi-Bellman (HJB) equation
\begin{equation}\label{hjb:planner}
  r v(x) = \max_{a} \left\{ - \frac{\phi}{2}a^2 + a v_x(x)
    + \frac{1}{2} \sigma^2 v_{xx}(x) \right\}
\end{equation}
subject to value matching ($v(\bar x) = b$,
$v(\underline{x}^{FB}) = \underline v$) and smooth pasting
($v_x(\underline{x}^{FB}) = 0$). The first-best effort in turn is pinned
down by the first-order condition
\begin{equation}\label{FB-FOC}
    a^{FB}(x) = \phi^{-1} v_x(x),
\end{equation}
which can be substituted back into \eqref{hjb:planner} to obtain a
simple ODE for $v$.

The following proposition summarizes properties of the first-best
solution.

\begin{proposition}\label{Prop:FirstBest}
  Let $v$, $a^{FB}$, and $\underline x^{FB}$ be the first-best value
  function, effort policy, and termination threshold. In the
  continuation region $[\underline x^{FB},\bar x]$, the following
  properties hold:
  \begin{enumerate}[label=(\roman*)]
  \item\label{it:d02} Value is increasing in progress: $v_x(x) \geq 0$
    with strict inequality if and only if $x> \underline x^{FB}$.
  \item\label{it:d03} Value is convex in progress: $v_{xx} > 0$.
  \item\label{it:d04} Effort is increasing in progress:
    $a^{FB}_x > 0$.
  \item\label{it:d05} Effort grows in expectation at rate $r$: the
    discounted effort process $e^{-rt} a^{FB}(X_t)$ is a martingale up
    to the first-best stopping time $\tau$.
  \end{enumerate}
\end{proposition}

\begin{figure}[t]
\centering
\begin{minipage}[c]{.48\linewidth}
\centering
\includegraphics[width=\linewidth]{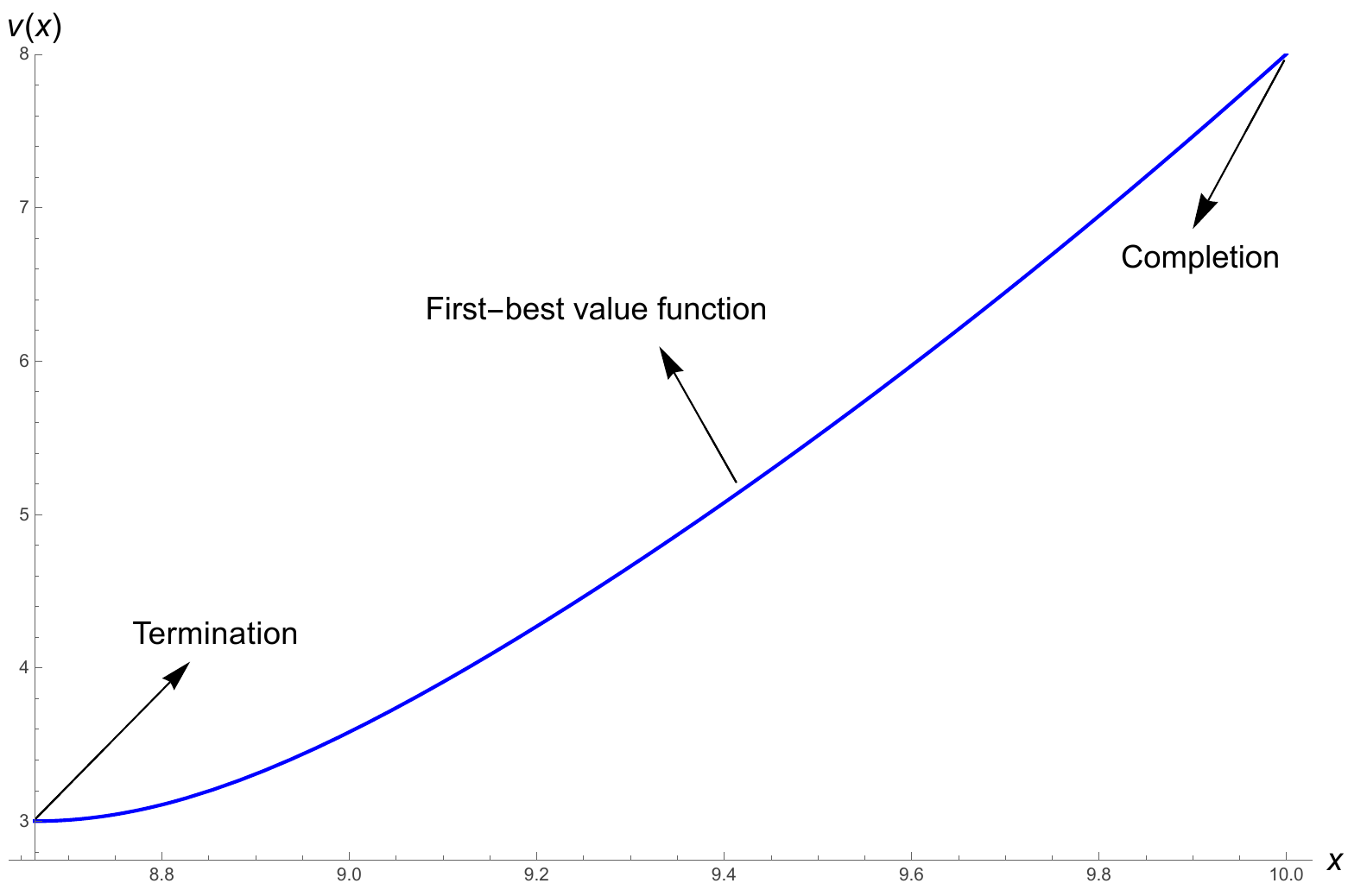}
\end{minipage}\quad
\begin{minipage}[c]{.48\linewidth}
\centering
\includegraphics[width=\linewidth]{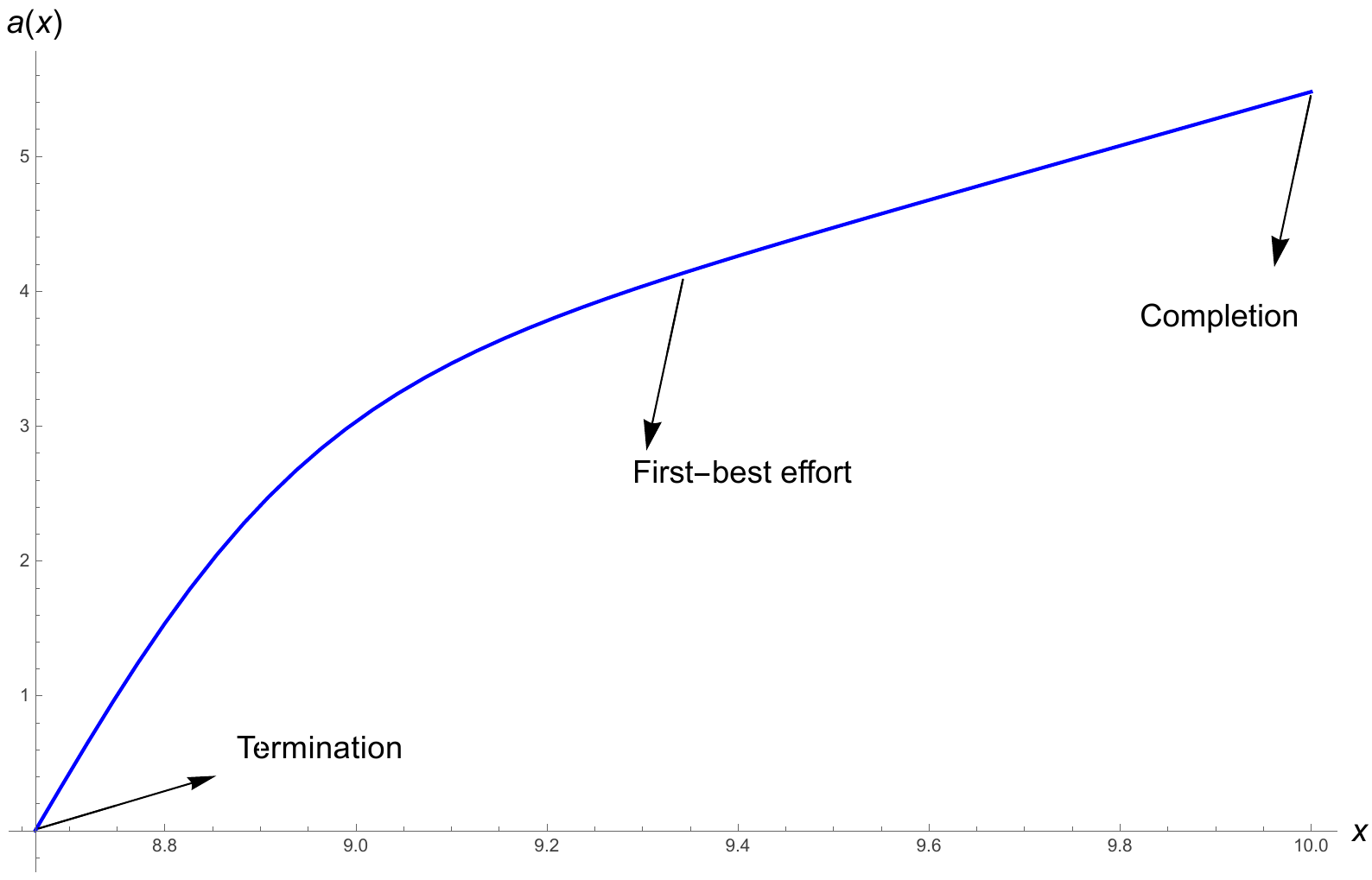}
\end{minipage}
\caption{First-best value function (left) and effort profile
  (right) for parameter values
  $(b, \bar x, \sigma, \underline v, r, \eta, \phi) = (8, 10, 1, 3, 2,
  2, 1)$. \label{fig:FirstBest}}
\end{figure}

The first best has an intuitive structure. If progress is sufficiently
far from the target, the project is terminated, as completing it would
be too costly. At the termination threshold, zero effort is
implemented. Everywhere else effort is positive and increasing in
progress. This arises from the payoff structure in the model: The cost
of exerting effort is incurred continuously, while the reward is
received only upon project completion. Therefore, the principal has a
stronger incentive to accelerate progress as the project nears
completion. The same logic operates along the path. Since progress
yields no payoff until the project is stopped, the envelope theorem
implies that the marginal value of progress today equals the expected
discounted marginal value of progress at any later date. It thus grows
in expectation at the discount rate, and so does effort, which is
proportional to it by \eqref{FB-FOC}. \autoref{fig:FirstBest}
illustrates the value function and the effort policy.

Next, we establish some comparative statics for the first-best solution.

\begin{proposition}\label{Prop:compfb}
  Consider two projects, labeled 1 and 2, that differ in exactly one
  parameter.
\begin{enumerate}[label=(\roman*)]

\item \label{it:pcfb01} If $b_1 > b_2$, then
  $\underline x^{FB}_{1} < \underline x^{FB}_{2}$. Furthermore,
  $a^{FB}_1(x) > a^{FB}_2(x)$ for all $x \in  (\underline
  x^{FB}_1,\bar x]$.

\item\label{it:pcfb02} If $r_1 > r_2$, then
  $\underline x^{FB}_{1} > \underline x^{FB}_{2}$. Furthermore, there
  exists an interior threshold $x^{\dagger}_r$ such that
  $a^{FB}_1(x) > a^{FB}_2(x)$ if $x \in (x^{\dagger}_r, \bar x]$
  whereas $a^{FB}_1(x) < a^{FB}_2(x)$ if
  $x \in (\underline x^{FB}_2, x^{\dagger}_r)$.

\item\label{it:pcfb03} If $\sigma_1 > \sigma_2$, then
  $\underline x^{FB}_{1} < \underline x^{FB}_{2}$. Furthermore, there
  exist interior thresholds
  $\{x^{\dagger}_{\sigma}, x^{\dagger \dagger}_{\sigma}\}$ such that
  $a^{FB}_1(x) > a^{FB}_2(x)$ if
  $x \in (\underline x^{FB}_1,x^{\dagger}_{\sigma})$ whereas
  $a^{FB}_1(x) < a^{FB}_2(x)$ if $x \in (x^{\dagger
    \dagger}_{\sigma},\bar x]$.

\item\label{it:pcfb04} If $\phi_1 > \phi_2$, then
  $\underline x^{FB}_{1} > \underline x^{FB}_{2}$. Furthermore,
  $a^{FB}_1(x) < a^{FB}_2(x)$ for all
  $x \in (\underline x^{FB}_2, \bar x]$.

\end{enumerate}
\end{proposition}

Part (i) is intuitive: the higher the benefit, the more valuable the project.

Part (ii) states that a less patient principal terminates the project
earlier than a more patient one. Further, the less patient principal
asks for more effort if and only if the project is sufficiently close
to completion. This is because she values early completion more highly
than a more patient principal. Since the reward is realized only upon
completion, and additional effort is costly upfront, the net benefit
of accelerating completion outweighs its cost only when the project is
near the threshold.

Part (iii) shows that an increase in the volatility $\sigma$ lowers
the termination threshold. This is because completing the project is a
real option, whose value increases with volatility. This part also
shows that incentives for effort become stronger with volatility when
the project is far from completion, and weaker as it nears
completion. To see this, observe that when the project is near the
termination threshold, the downside is negligible since
$v_{x}(x) \approx 0$ and $v(x) \approx \underline v$. In contrast, the
upside is significant because $v_{xx}(\underline x) > 0$. The opposite
occurs as the project approaches the completion point $\bar x$.

Finally, part (iv) shows that higher effort costs lead to earlier
termination and lower effort. The intuition is direct: with effort
more expensive, less is exerted at every progress level and the
project is abandoned sooner as completing it would be more expensive.

\section{Unobservable Effort} \label{Sec:Analysis}

In this section, we begin the analysis of the optimal contracting
problem under moral hazard. The relevant state variables for the
problem are the project's current progress level and the agent's
continuation utility. We first note some necessary conditions for
incentive compatibility and use them to formulate the principal's
problem recursively. We then characterize the optimal contract.

\subsection{Local Incentive Compatibility}\label{Sec:LocalIC}

Given an incentive-compatible contract ($A,C,\tau)$, the agent's
continuation utility is defined by
\begin{equation*}
W_t \ceq \mathbb{E}^A_t \left[ \int_t^\infty e^{-r (s -t)} u(A_s,C_s) ds \right].
\end{equation*}
Its evolution takes the following familiar form.

\begin{lemma}[\citealp{sannikov2008continuous, he2011model}]\label{Lem:W}
  The agent's continuation utility under any incentive-compatible
  contract $(A,C,\tau)$ satisfies
  \begin{equation}\label{W-SDE}
    dW_t = \left( r W_t - u(A_t,C_t) \right) dt
    + \beta_t (-\eta r W_t)(dX_t - A_t dt)
  \end{equation}
  for some progressively measurable process $\beta$.
\end{lemma}

The proof is standard and hence omitted. We have adopted
\possessivecite{he2011model} convention of normalizing the process
$\beta$ controlling the strength of incentives by $-\eta r W_t$. (Note
that $W_t < 0$ because of CARA utility.) This way, $\beta_t$ measures
directly the sensitivity of the agent's monetary compensation to his
performance.

The following lemma gives necessary conditions for incentive
compatibility.

\begin{lemma}[\citealp{he2011model}]\label{Lem:IcNoS}
  Let $(A,C,\tau)$ be an incentive-compatible contract and let $\beta$
  be the process from \eqref{W-SDE}. Then $r W_t = u(A_t,C_t)$ and
  thus the agent's continuation utility $W_t$ is a
  martingale. Moreover, the agent's consumption is given by
\begin{equation}\label{cons}
C_t = \frac{\phi}{2}A_t^2 + r\, CE(W_t),
\end{equation}
and effort is given by
\begin{equation}\label{IC}
A_t = \phi^{-1}\max\{\beta_t, 0\}.
\end{equation}
\end{lemma}

The above conditions follow readily from first-order conditions to the
agent's problem as in \cite{he2011model}. To sketch the argument,
suppose the agent has deviated under the incentive-compatible contract
$(A,C,\tau)$ and has accumulated private savings $S_t=s \in \R$ at
time $t$. Then his continuation utility, $W_t(s)$, depends only on the
public history of progress (as the contract conditions on it) and the
savings $s$. Using the properties of CARA utility, it is
straightforward to verify that $W_t(s)= e^{-\eta r s}W_t$, where
$W_t= W_t(0)$ is the continuation utility at this same progress
history under no savings (see Lemma~\ref{lem:AWE} in the Appendix). This is because under CARA, the optimal use of the savings
$s$ is to simply consume the interest $rs$ at every time instant in
perpetuity, and the corresponding constant part of the consumption
utility can be factored out of the agent's problem. This implies that,
holding effort fixed, the Euler equation for optimal time-$t$
consumption takes the form
\begin{equation*}
  u_c(A_t,C_t) = \frac{d}{ds}W_t(s) = -\eta r e^{-\eta r s}W_t.
\end{equation*}
Noting that $u_c(A_t,C_t) = - \eta u(A_t,C_t)$ by \eqref{carau} and
that $s=0$ by incentive compatibility, this Euler equation simplifies
to $u(A_t,C_t) = rW_t$, which is the first claim in
Lemma~\ref{Lem:IcNoS}. Using CARA utility and \eqref{eq:CE}, this is
seen to be equivalent to \eqref{cons}.

The effort condition \eqref{IC} in turn is simply the condition for
$A_t$ to be optimal in the agent's one-shot-deviation problem
\begin{equation*}
  \max_{a \geq 0} a \beta_t(-\eta r W_t) +u(a,C_t),
\end{equation*}
where the first term captures the expected impact of effort on
continuation utility and the second term captures the flow disutility
of effort (cf.\ \citealp{sannikov2008continuous}). If
$\beta_t \geq 0$, the first-order condition evaluated at $a=A_t$ gives
$\beta_t (\eta r W_t) = u_a(A_t,C_t)$. By CARA utility,
$u_a(A_t,C_t)= u_c(A_t,C_t)(-\phi A_t) = \eta r W_t \phi A_t$, where
the second equality follows by the consumption Euler evaluated at
$s=0$. Thus, $\beta_t = \phi A_t$. If instead $\beta_t<0$, then
$A_t=0$.

The upshot of Lemma~\ref{Lem:IcNoS} for the principal's problem is
that the level of the agent's continuation utility cannot play a role
in providing incentives. By \eqref{IC}, the sensitivity $\beta_t$ of
the agent's monetary compensation to performance pins down his effort
$A_t$ independent of $W_t$, as expected under CARA
preferences. Furthermore, because the agent can freely save and
borrow, toying with the timing of when $W_t$ is delivered does not
change what the agent does. One way to deliver $W_t$ is by using an
incentive-compatible contract under which the agent does not save or
borrow, in which case the payments to him are additively separable in
the cost of effort and the interest on the certainty equivalent of
$W_t$ as shown in equation \eqref{cons}. The agent's continuation
utility $W_t$ then enters the principal's value function as an
additive cost term; we verify this formally in the next section.

\subsection{Recursive Formulation}\label{Sec:Recursive}

We consider a relaxation of the second-best problem
\eqref{second-best} where incentive compatibility is replaced with the
local conditions from Lemma~\ref{Lem:IcNoS}. That is, at any time $t$,
given current progress $x \leq \bar x$ and continuation utility $w$,
we solve
\begin{equation}\label{sb-sequence}
  f(x,w) \ceq \max_{A, C, \tau}
  \mathbb{E}_{t,(x,w)}^{A}\bigg[ e^{-r (\tau-t)}
  \big(\mathbf{1}_{\{X_\tau \geq \bar{x}\}} b
  + \mathbf{1}_{\{X_{\tau} < \bar{x}\}} \underline v \big)
  - \int_{t}^{\infty} e^{-r (s-t)} C_s\, ds \bigg]
\end{equation}
subject to the laws of motion \eqref{progress} and \eqref{W-SDE} with
$(X_t,W_t)= (x,w)$ as well as the local incentive compatibility
conditions \eqref{cons} and \eqref{IC}. Then, by construction,
$f(0,CE^{-1}(0))$ is an upper bound on the principal's maximal payoff
in \eqref{second-best}.

To formulate the relaxed problem recursively, we guess that the value
function $f$ is additively separable with the following functional
form:
\begin{equation}\label{fdefn}
f(x,w) = \pi(x) - CE(w).
\end{equation}
We refer to $\pi(x)$ as the project value; the second term is the cost
of delivering the continuation utility $w$. With this guess, the HJB
equation for the relaxed problem is
\begin{equation*}
  rf(x,w) = \max_{a,c,\beta} \Big\{ -c + a f_x(x,w) 
    + \frac{1}{2} \sigma^2 f_{xx}(x,w)
    + \frac{1}{2} [\beta \eta r (-w) \sigma]^2 f_{ww}(x,w) \Big\},
\end{equation*}
because the drift of continuation utility is zero by
Lemma~\ref{Lem:IcNoS}, and $f_{xw} = 0$. Substituting for consumption
$c$ and sensitivity $\beta$ using \eqref{cons} and \eqref{IC} (with
$\beta = 0$ when $a = 0$, since $f_{ww} < 0$), and using the functional
form of $f$, this becomes
\begin{equation*}
  r \pi(x) - r\, CE(w) = \max_{a }
  \left\{ -\frac{\phi}{2} a^2 - r\, CE(w) + a
    \pi_x(x) + \frac{1}{2} \sigma^2 \pi_{xx}(x) - \frac{1}{2}
    \eta r \phi^2 a^2 \sigma^2 \right\},
\end{equation*}
which simplifies to
\begin{equation}
  r \pi(x) = \max_{a } \left\{ -\frac{\phihat}{2} a^2
    + a \pi_x(x) + \frac{1}{2} \sigma^2 \pi_{xx}(x) \right\},\label{prpr}
\end{equation}
where we have defined $\kappa \ceq 1 + \phi \eta r \sigma^2$ and
$\phihat \ceq \phi \kappa$.

Note that equation \eqref{prpr} only involves the project value $\pi$,
with progress the only state variable. By inspection of
\eqref{hjb:planner}, it is of the same form as the first-best HJB, the
only difference being the adjusted cost parameter $\hat \phi$. It
follows that here, too, the solution consists of a termination
threshold $\underline x$ and a value function $\pi$ which is
infinitely differentiable on the continuation region
$[\underline x, \bar x]$ and satisfies the HJB equation \eqref{prpr}
subject to value matching
\begin{equation}
\label{Val}
\pi(\bar x) = b, \quad \pi(\underline x) = \underline v,
\end{equation}
and smooth pasting
\begin{equation}
\label{past}
\pi_x(\underline x) = 0.
\end{equation}
This solution then defines the value $f$ of the relaxed problem
\eqref{sb-sequence} via \eqref{fdefn}.

We show next that the bound is tight: $f(0,CE^{-1}(0))= \pi(0)$ is the
value of the principal's problem \eqref{second-best}, and $\pi$ is the
optimal project value. To this end, define a contract $(A, C, \tau)$
from the solution to the HJB equation \eqref{prpr} subject to
\eqref{Val} and \eqref{past} as follows. Let $\tau$ be the first time
$X_t \notin (\ux, \bar x)$. For $t \leq \tau$, set
$A_t = \hat\phi^{-1}\pi_x(X_t)$ and $\beta_t = \phi A_t$, and
determine $C_t$ and $W_t$ via \eqref{W-SDE} and \eqref{cons}. For
$t > \tau$, set $A_t = 0$ and $C_t = r\, CE(W_\tau)$.

\begin{proposition}\label{Prop:BaseSuff}
  Let $(A,C,\tau)$ be the contract defined by the solution to
  \eqref{prpr} subject to \eqref{Val}--\eqref{past}. Then $(A,C,\tau)$
  is incentive compatible, and thus it is an optimal contract.
\end{proposition}

Proposition~\ref{Prop:BaseSuff} is a special case of
Proposition~\ref{Prop:MartingaleLevy} in Section~\ref{Sec:DisModel}.

\subsection{The Optimal Contract} By Proposition~\ref{Prop:BaseSuff},
an optimal contract can be found by solving the auxiliary planner's
problem \eqref{prpr}. This problem differs from the first-best problem
in Section~\ref{Sec:Planner} only in that the cost parameter $\phi$ is
replaced by the incentive-adjusted cost parameter
$\hat \phi = \kappa \phi > \phi$, where
$\kappa = 1+ \phi \eta r \sigma^2$.\footnote{The reduction to an
  auxiliary planner's problem is not specific to quadratic costs. For
  a general effort cost $h$, the sensitivity required to induce effort
  $a$ is $\beta = h_a(a)$, and the incentive-adjusted cost is
  $\hat h(a) \ceq h(a) + \frac{\eta r \sigma^2}{2} h_a(a)^2$ (cf.
  \citealp[eq.~(13)]{he2011model}), which reduces to the
  multiplicative form $\hat h = \kappa h$ when $h$ is quadratic. Our
  results carry over to a general $h$ as long as $\hat h$ is more
  convex than $h$ in the sense that $\hat h_{aa} \geq h_{aa}$, which
  is ensured by $h_{aaa} \geq 0$ (see footnote~\ref{fn:cost}).} The
adjustment is a risk premium. Inducing effort $a$ requires tying the
agent's compensation to progress with sensitivity $\beta = \phi a$
(Lemma~\ref{Lem:IcNoS}). Compensating the CARA agent for the resulting
flow variance $\beta^2 \sigma^2$ costs
$\frac{\eta r}{2} \beta^2 \sigma^2$ per unit of time; the coefficient
of risk aversion is $\eta r$ rather than $\eta$ because the agent
spreads any change in wealth over his future consumption.\footnote{See
  Lemma~\ref{lem:AWE} in the Appendix.} Hence, the incentive-adjusted
cost
\begin{equation*}
  \frac{\phihat}{2} a^2 =\frac{\phi}{2} a^2 + \frac{\eta r}{2} \beta^2 \sigma^2
\end{equation*}
is simply the cost of effort plus the risk premium for the exposure
needed to induce it. The upshot is that the effect of dynamic moral
hazard on effort and the principal's value has thus been reduced to a
comparative static with respect to $\phi$ in the first-best
problem. (The agent's compensation and consumption of course differ
qualitatively from the first best as he now has skin in the game.)

The first-order condition to problem
\eqref{prpr} gives the optimal effort
\begin{equation*}
a(x) = \phihat^{-1} \pi_x(x).
\end{equation*}
By Lemma~\ref{Lem:IcNoS}, the pay-performance sensitivity required to induce it is
\begin{equation}\label{eq:beta}
  \beta(x) = \phi a(x) = \frac{\pi_x(x)}{1+\phi \eta r \sigma^2}.
\end{equation}
Thus, at each progress level, incentives are exactly as in the static
linear-CARA-normal model of \citet{holmstrom1987aggregation} for a
project with marginal product $\pi_x(x)$ and an agent with
risk-aversion coefficient $\eta r$. Substituting the first-order
condition back into the HJB equation \eqref{prpr} yields a
second-order ODE for the project value $\pi$:
\begin{equation}\label{eq:ODE}
r \pi(x) = \frac{1}{2\phihat} \pi_x(x)^2 + \frac{1}{2} \sigma^2 \pi_{xx}(x).
\end{equation}
The problem thus boils down to finding a termination threshold
$\underline x$ and a project value $\pi$ that solve this ODE subject
to value matching \eqref{Val} and smooth pasting \eqref{past}.

The following proposition summarizes the optimal contract.

\begin{proposition}\label{Prop:SecondBest}
  Let $\pi$, $a$, and $\ux$ be the project value, effort policy, and
  termination threshold under the optimal contract, and let $\beta$ be
  the pay-performance sensitivity. In the continuation region
  $[\ux, \bar x]$, the following properties hold:
  \begin{enumerate}[label=(\roman*)]
  \item\label{it:s01} Value is increasing in progress:
    $\pi_x(x) \geq 0$ with strict inequality if and only if $x > \ux$.
  \item\label{it:s02} Value is convex in progress:
    $\pi_{xx} > 0$.
  \item\label{it:s03} Effort and incentives are increasing in
    progress: $a_x > 0$ and $\beta_x > 0$.
  \item\label{it:s04} Effort grows in expectation at rate $r$: the
    discounted effort process $e^{-rt} a(X_t)$ is a martingale up to
    the second-best stopping time $\tau$.
  \end{enumerate}
\end{proposition}

Since the optimal effort and project value are derived from the
auxiliary planner's problem, the result and the intuition are
analogous to Proposition~\ref{Prop:FirstBest}.

\paragraph{A Closed-Form Example} When $\uv = \phihat \sigma^2/2$, the
project value ODE can be solved in closed form to illustrate
Proposition~\ref{Prop:SecondBest}. In this case, value is quadratic in
progress, with
\begin{equation*}
  \pi(x) = \uv + \frac{r \phihat}{2} (x - \ux)^2, \qquad
  \ux = \bar x - \sqrt{\frac{2(b - \uv)}{r \phihat}}.
\end{equation*}
Effort and incentives are linear in the distance from the termination
threshold:
\begin{equation*}
  a(x) = r (x - \ux), \qquad \beta(x) = \phi r (x - \ux).
\end{equation*}
Since $a(X_t) = r(X_t - \ux)$, effort has drift $r\, a(X_t)$ along the
path. \hfill $\square$

\bigskip

The next proposition provides comparative statics of the optimal
contract.

\begin{proposition}\label{Prop:comp}
  Consider two projects, labeled 1 and 2, that differ in exactly one
  parameter.
\begin{enumerate}[label=(\roman*)]
\item\label{it:pc01} If $b_1 > b_2$, then $\ux_1 < \ux_2$ and
  $a_1(x) > a_2(x)$ for all $x \in (\ux_1, \bar x]$.
\item\label{it:pc02} If $\eta_1 > \eta_2$, then $\ux_1 > \ux_2$ and
  $a_1(x) < a_2(x)$ for all $x \in (\ux_2, \bar x]$.
\item\label{it:pc03} If $\phi_1 > \phi_2$, then $\ux_1 > \ux_2$ and
  $a_1(x) < a_2(x)$ for all $x \in (\ux_2, \bar x]$.
\item\label{it:pc04} If $r_1 > r_2$, then $\ux_1 > \ux_2$, and there
  exists $x^{\ddagger} \in (\ux_1, \bar x]$ such that $a_1(x) < a_2(x)$
  for all $x \in (\ux_2, x^{\ddagger})$ and $a_1(x) > a_2(x)$ for all
  $x \in (x^{\ddagger}, \bar x]$.
\end{enumerate}
\end{proposition}

The benefit $b$ affects the optimal contract as it does the first
best. Risk aversion $\eta$ and the cost parameter $\phi$ enter the
second-best problem only through the incentive-adjusted cost parameter
$\phihat$, which is increasing in each, so a more risk-averse agent is
treated like one with a higher cost of effort: the principal tolerates
less failure and induces less effort. The discount rate $r$ enters
twice, as the rate proper and through $\phihat$. Both channels make
the principal less tolerant of failure, but they conflict for effort:
impatience raises effort near the target
(Proposition~\ref{Prop:compfb}\ref{it:pcfb02}), whereas higher
$\phihat$ lowers it everywhere, so the contract with the higher
discount rate implements less effort near the threshold and either
less or more near the target. Volatility is absent from the
proposition because it, too, enters twice, and the two channels
conflict already for the threshold: a higher $\sigma$ raises the
option value of continuing but also the cost of incentives, and either
may dominate.

\begin{figure}[t]
\centering
    \includegraphics[width=.48\linewidth]{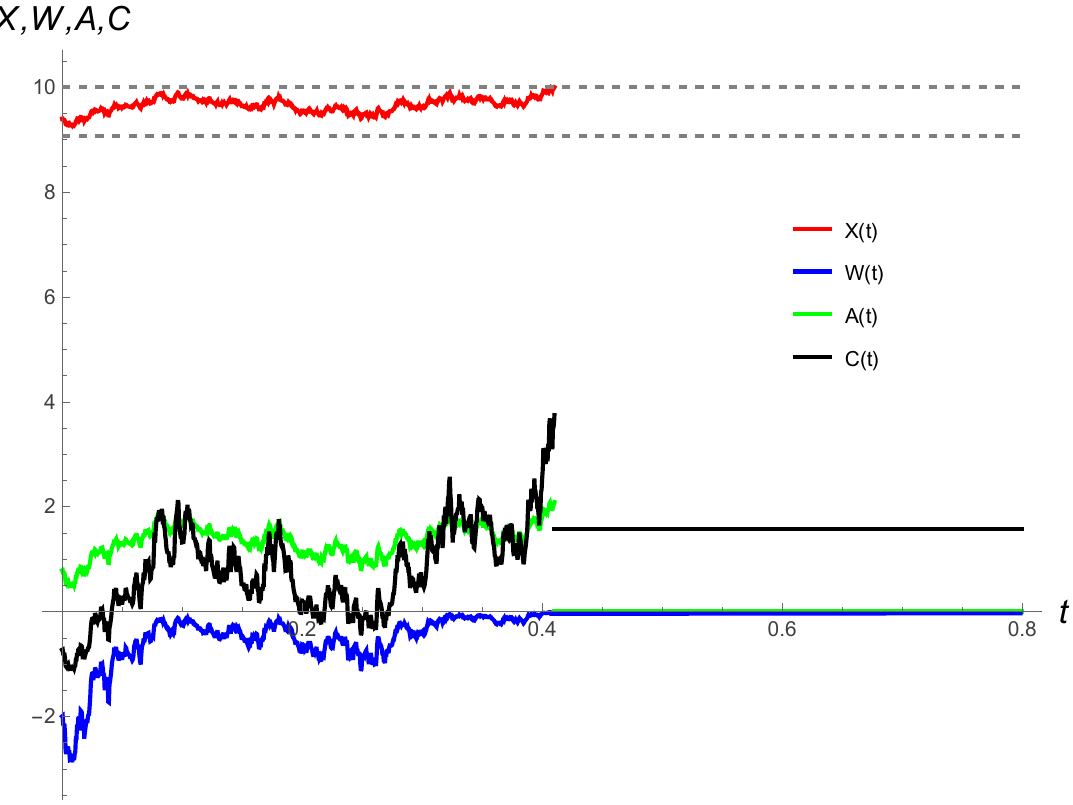}\quad
    \includegraphics[width=.48\linewidth]{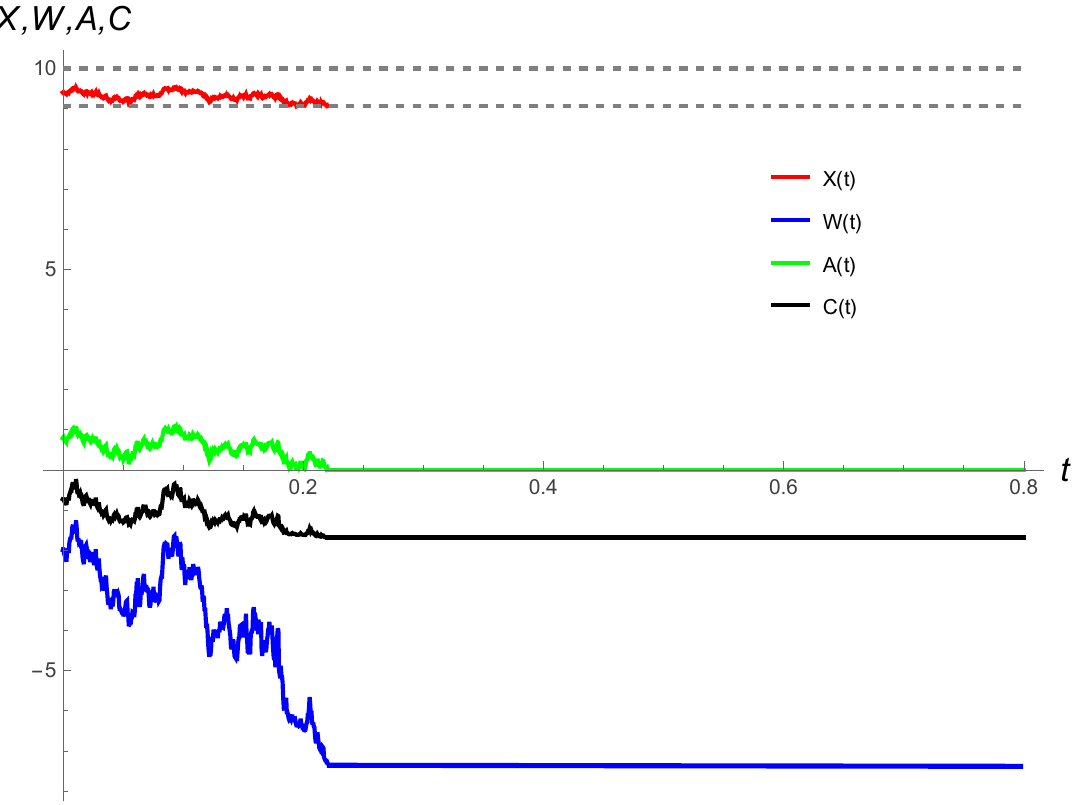}
\caption{Sample paths for a completed project (left panel) and a terminated project (right panel). Parameter values: $(b,\bar x,\sigma,\underline v,r,\eta,\phi,X_0 )=(8,10,1,3,2,2,1,9.4)$. }\label{fig:Sample}
\end{figure}

Figure~\ref{fig:Sample} plots two sample paths of progress under the
optimal contract, along with the agent's continuation utility, effort,
and consumption (which by incentive compatibility coincides with
compensation).  In the left panel the project is completed; in the
right panel it is terminated early. Along the successful path, effort
and consumption rise as progress nears the target. At completion,
effort drops to zero and consumption drops by the cost of effort and
stays constant thereafter. This downward jump does not arise in the
unsuccessful project, since effort converges to zero as the
termination threshold approaches, so that consumption converges
continuously to its post-termination level.

Note that, as the project evolves, the optimal contract remains
renegotiation-proof: the principal's full value
$f(X_t,W_t) =\pi(X_t) - CE(W_t)$ is decreasing in the
agent's continuation value $W_t$ at every $(X_t,W_t)$, so at no
history can both parties gain by replacing the contract (cf.\
\citealp{fudenberg1990short}).

\begin{figure}[t]
\centering
\begin{minipage}[c]{.48\linewidth}
\centering
\includegraphics[width=\linewidth]{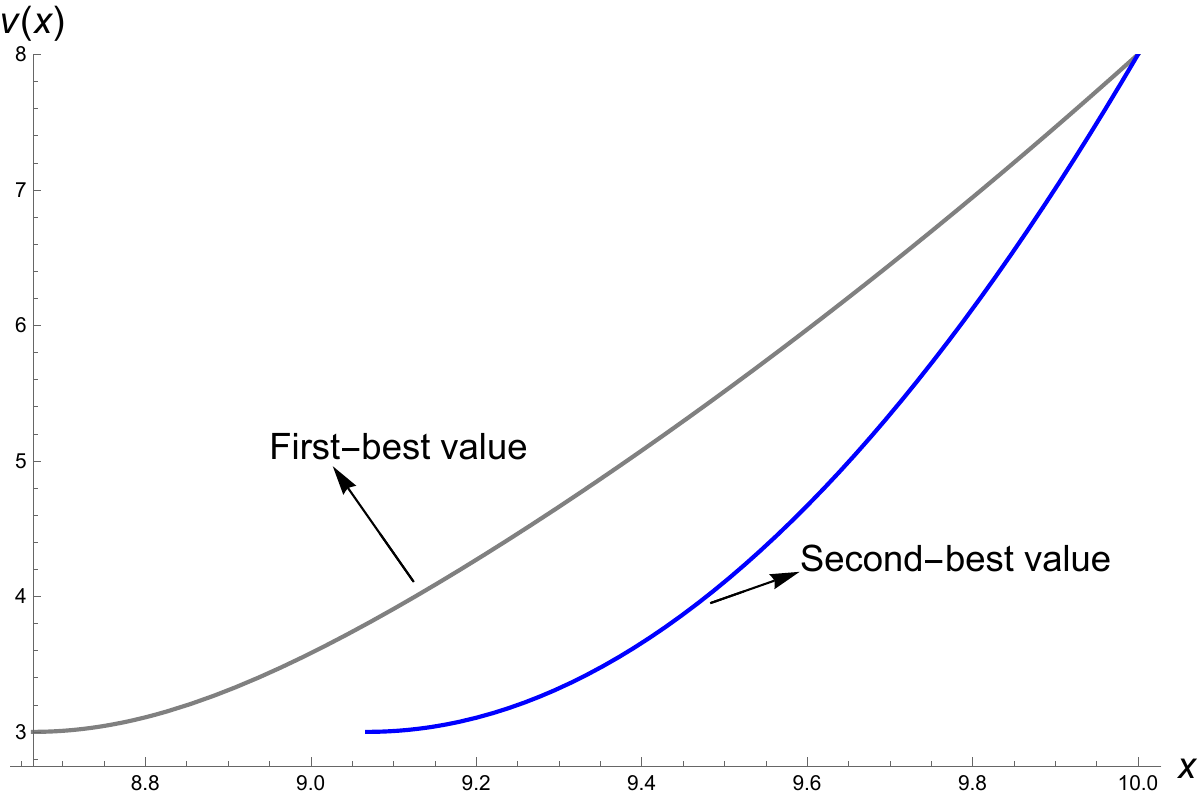}
\end{minipage}\quad
\begin{minipage}[c]{.48\linewidth}
\centering
\includegraphics[width=\linewidth]{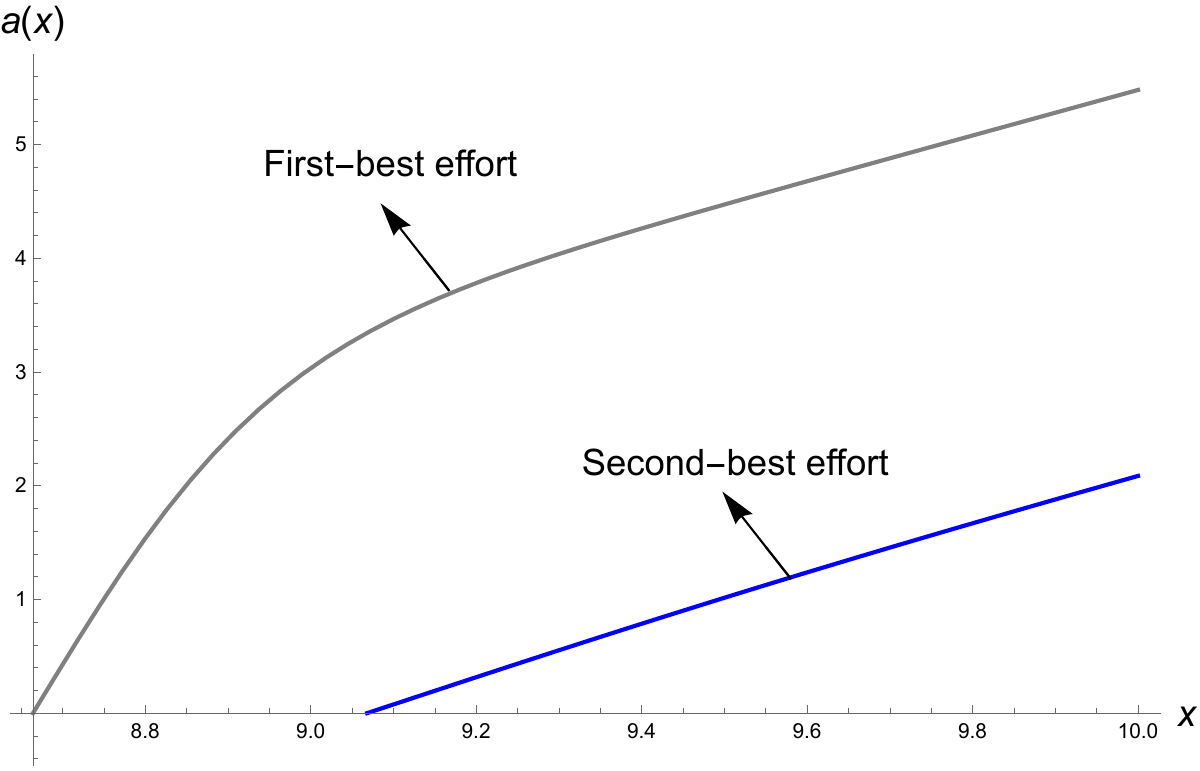}
\end{minipage}
\caption{Second-best value function (left) and effort profile (right) compared to the first best. Parameter values: $(b,\bar x,\sigma,\underline v,r,\eta,\phi )=(8,10,1,3,2,2,1)$. \label{fig:Principal}}
\end{figure}

We then compare the optimal contract to the first best.

\begin{proposition}\label{Prop:Principal}
  Let $\pi$, $\ux$, and $a$ be the project value, termination
  threshold, and effort policy under the optimal contract, and let
  $v$, $\ux^{FB}$, and $a^{FB}$ be their first-best counterparts.
\begin{enumerate}[label=(\roman*)]
\item\label{it:p01} Value is lower under moral hazard, with the
  difference vanishing with the risk premium:
  \begin{equation*}
    0 < v(x) - \pi(x) \leq \phi \eta r \sigma^2\,(b - v(x)) \quad
    \text{for all } x \in (\ux^{FB}, \bar x).
  \end{equation*}
\item\label{it:p02} Moral hazard makes the principal less tolerant of
  failure: $\ux > \ux^{FB}$.
\item\label{it:p03} Effort is uniformly lower under moral hazard:
  $a(x) < a^{FB}(x)$ for all $x \in [\ux, \bar x]$.
\end{enumerate}
\end{proposition}

\autoref{fig:Principal} plots the second-best value function and
effort profile against their first-best counterparts. Since the
second-best problem is the first-best problem with $\phihat > \phi$,
the comparisons follow as comparative statics of the first-best
problem. The upper bound on the value gap in turn obtains by
implementing the first-best policy under moral hazard. It shows that
the cost of moral hazard vanishes if effort costs vanish
($\phi \to 0$), the agent approaches risk neutrality ($\eta \to 0$),
noise vanishes ($\sigma^2 \to 0$), or the parties become perfectly
patient ($r \to 0$) making waiting free, with both policies relying on
mere luck.

The ranking of efforts and thresholds has a sharp implication for
progress itself, which can be seen by using a coupling argument. 

\begin{proposition}\label{Prop:Coupling}
  Let $X^{FB}$ and $X$ denote progress under the first-best and the
  second-best policy when both are driven by the same Brownian motion,
  with $\tau^{FB}$ and $\tau$ the corresponding stopping
  times. Suppose the project is worth starting in the first
  best. Then, almost surely, ${X^{FB}_{t \wedge \tau^{FB}} > X_t}$ for
  all $0 < t < \tau$. Consequently, whenever the second best completes
  the project, the first best has completed it strictly earlier;
  whenever the first best terminates the project, the second best has
  terminated it strictly earlier; and the probability of completion is
  higher under the first best.
\end{proposition}

That is, moral hazard not only makes the project more likely to fail,
it also extends the completion time for any realized path of the
Brownian shocks. In particular, if we identify termination with
completion at $+\infty$, then for any path of shocks, the first-best
project is completed weakly earlier than the second best, and strictly
earlier whenever the first-best completion time is finite. Conversely,
identifying completion with termination at $+\infty$, the second best
leads to weakly earlier termination than the first best along every
realized shock path, strictly so when the second-best termination time
is finite. Project durations are not ranked, however, if we do not
condition on the outcome: the project can end earlier under the second
best, when it is terminated while the more failure-tolerant first-best
policy still continues, or later, when the first best has already
completed it while the second best is still at work.

Finally, we show that agency frictions may distort the goal setting
and design of projects in addition to their execution. To do so,
assume that the principal can choose the scope of the
project. Specifically, instead of the target $\bar x$ and the benefit
$b$ being given, a project completed at progress level $x$ yields a
benefit $b(x)$, where $b$ is increasing and concave with
$b(0) < \uv < \lim_{x \to \infty} b(x)$ and $b_x(x) \to 0$ as
$x \to \infty$. The principal can then end the project at any time
$\tau$ and either abandon it for the salvage value $\uv$ or declare it
complete for the benefit $b(X_\tau)$. We assume that the project is
worth starting. The solution again takes the form of a project value
function $\pi$ that solves the project value ODE \eqref{eq:ODE} on the
continuation region $[\ux, \bar x]$, with the project terminated at
$\ux$ and completed at $\bar x$, but now both thresholds are
endogenous. Thus, value matching and smooth pasting must hold at both
ends: $\pi(\ux) = \uv$, $\pi_x(\ux) = 0$, $\pi(\bar x) = b(\bar x)$,
and $\pi_x(\bar x) = b_x(\bar x)$ (see the Supplemental Appendix for
details). The first-best problem is analogous, with $\phi$ in place of
$\phihat$; its continuation region is $[\ux^{FB},\bar x^{FB}]$.

\begin{proposition}\label{Prop:Endo}
  With endogenous scope, the optimal contract is less ambitious and
  less tolerant of failure than the first best: $\bar x < \bar x^{FB}$
  and $\ux > \ux^{FB}$.
\end{proposition}

Whatever the choice of the target $\bar x$, the principal then offers
the agent the optimal contract from Proposition~\ref{Prop:BaseSuff}
for this $\bar x$.

\subsection{Implementation}\label{Sec:Implementation}

The optimal contract (with exogenous or endogenous project scope) pins
down the agent's consumption, but the payments are unique only up to
present value along each path. We say that a payment scheme implements
the optimal contract if, given the stopping rule $\tau$, it is optimal
for the agent to follow the optimal effort policy $A_t = a(X_t)$, and
this yields both parties the same expected payoffs as the optimal
contract. We note here three natural implementations. The first is the
incentive-compatible contract itself, under which the agent consumes
what he is paid: by \eqref{cons},
$C_t = \frac{\phi}{2} A_t^2 + r\, CE_t$, where $CE_t \ceq CE(W_t)$ is
the certainty equivalent of his continuation utility.

Second, all payments can be deferred until the project is
stopped. Paying at $\tau$ the lump sum
$\int_0^\tau e^{r(\tau - t)} C_t\, dt + CE_\tau$, the compounded value
of the interim payments plus the certainty equivalent of the remaining
promise, leaves the present value of payments (and hence the agent's
feasible set) unchanged along every path. In this implementation, the
agent either borrows against the lump sum payment to finance
consumption along the way, or, after poor performance, saves in
anticipation of a small or even negative payment---a penalty akin to
liquidated damages or a malus on deferred pay. Payment then occurs
only upon completion or termination, though the amount still depends
on the path of progress, and not just the stopping time.

Third, the optimal contract can be implemented via a sequence of spot
contracts (cf.\ \citealp{fudenberg1990short}). By It\^o's formula,
$dCE_t = \beta_t (dX_t - A_t\, dt) + \frac{\eta r}{2} \beta_t^2
\sigma^2\, dt$. Consider the payment scheme $\Gamma$ that pays at each
$t \leq \tau$ (i.e., while the project is active) the flow
\begin{equation}\label{eq:spot}
  d\Gamma_t \ceq \Big[\frac{\phi}{2} A_t^2 + \frac{\eta
    r}{2} \beta_t^2 \sigma^2 \Big] dt + \beta_t\,(dX_t - A_t\, dt),
\end{equation}
with $A_t = a(X_t)$ and $\beta_t = \beta(X_t)$, and nothing
thereafter. By construction, we have
$d\Gamma_t = \frac{\phi}{2} A_t^2\, dt + dCE_t$, so by inspection of
\eqref{cons}, the scheme pays the increment $dCE_t$ of the agent's
certainty equivalent in place of its annuity $r\, CE_t\, dt$, leaving to
the agent the consumption smoothing the incentive-compatible contract
performs on his behalf.

\begin{proposition}\label{Prop:Spot}
  Under the spot payment scheme $\Gamma$, for every $t \geq 0$, almost
  surely, the agent's continuation problem at $t$ with any savings
  $S \in \R$ coincides with his continuation problem under the optimal
  incentive-compatible contract with savings $S - CE_t$. It follows
  that his continuation utility has certainty equivalent $S$, the
  value of his outside option. Hence, $\Gamma$ implements the optimal
  contract, with the agent consuming $C_t$ and holding savings $CE_t$,
  and with the participation constraint satisfied at all times, on
  and off the path.\footnote{Formally,
    letting $W_t(S; \Gamma)$ denote the agent's time-$t$ continuation
    utility under $\Gamma$ given savings $S$, the participation
    constraint is $CE(W_t(S; \Gamma)) \geq S$ almost surely, the
    outside option being worth $S$ (Lemma~\ref{lem:AWE}).}
\end{proposition}

The idea is simple: rather than smoothing the agent's consumption, the
principal can just pay him what he is owed as it accrues, and the
agent smooths consumption himself by consuming only the interest on
his savings. That is, whereas under the incentive-compatible contract
the agent's financial account is kept with the principal, the spot
implementation $\Gamma$ has the agent bank with the outside credit
market instead.

Remarkably, the spot implementation $\Gamma$ requires no commitment on
either side. Leaving with savings $S$ is worth $S$ to the agent
(Lemma~\ref{lem:AWE}). By Proposition~\ref{Prop:Spot}, this is exactly
the value of his continuation utility at every history and every
savings level, so the participation constraint holds at all times, on
and off the path. The scheme is likewise sequentially rational for the
principal: since the agent's outside option scales one to one with his
savings, an agent with savings $S$ and an outside option worth $S$ is
equivalent to an agent with no savings and an outside option worth
$0$, so the principal's continuation value given progress $X_t$ under
$\Gamma$ is $f(X_t, CE^{-1}(0)) = \pi(X_t)$, the value of the optimal
(long-term) contract starting from $X_t$ with an agent with no
rent. Thus, the principal cannot gain by altering the spot terms or
the stopping rule. Nor does she owe anything once the project stops:
the agent's post-termination consumption under the optimal contract
equals $r\, CE_\tau$, which is exactly what his savings finance.

It is also worth noting that each spot contract in $\Gamma$ is an
optimal contract in the static linear-CARA-normal model of
\citet{holmstrom1987aggregation}: by \eqref{eq:spot} and
\eqref{eq:beta}, it pays the agent his effort cost and a risk premium,
plus $\beta_t = \pi_x(X_t)/(1 + \phi \eta r \sigma^2)$ times the
innovation in progress, which is their contract for output $dX_t$ with
marginal product $\pi_x(X_t)$ and risk aversion $\eta r$. Importantly,
linearity is not assumed but it is shown to be fully optimal. However,
here the slope $\beta_t$ varies with progress, so the aggregate
compensation is not linear in cumulative progress.\footnote{In the
  stationary version of our model, in which the principal receives the
  flow $dX_t$ rather than a reward upon completion, the optimal
  contract is stationary, with $\beta = 1/\kappa$ and effort
  $1/\phihat$. In this version the linear spot contracts aggregate and
  the agent's compensation is linear in cumulative progress over any
  horizon, reproducing the dynamic solution in
  \citet{holmstrom1987aggregation} with $\eta r$ as the coefficient of
  risk aversion. This gives an infinite-horizon model with private
  intermediate consumption and savings in which linearity is fully
  optimal, with no need for limiting arguments (cf.\
  \citealp{hellwig2002discrete}).}

\section{Endogenous Resets}\label{Sec:ResModel}

In the baseline model, there is only a single viable approach to
reaching the target, and abandoning it thus amounts to terminating
the entire project. Here, we analyze the opposite extreme, in which an
infinite number of ex ante identical approaches is available, and the
principal can switch to a fresh one instead of terminating the
project.

Specifically, the principal can reset the project at any time by
incurring a fixed cost $L > 0$, upon which progress restarts from an
exogenous restart point $x^R < \bar x$. We allow for $x^R \neq 0$, so
that the original approach may be superior (or inferior) to the
others.\footnote{Whether the agent is retained or replaced at a reset
  is immaterial: under CARA the cost of his continuation utility is
  the same either way, and under the spot implementation nothing is
  owed.}

A contract now also specifies the sequence of stopping times
$T_1 \leq T_2 \leq \cdots$ at which progress is reset to
$x^R$. The principal's problem is \eqref{second-best} with the
reset costs subtracted:
\begin{equation}\label{reset-problem}
  \max_{A, C, \tau, (T_i)} \mathbb{E}^{A} \bigg[ e^{-r \tau}
  \big(\mathbf{1}_{\{X_\tau \geq \bar{x}\}} b
  + \mathbf{1}_{\{X_{\tau} < \bar{x}\}} \underline v \big)
  - \int_0^{\infty} e^{-r t} C_t \, dt
  - \sum_{i} \mathbf{1}_{\{T_i \leq \tau\}} e^{-r T_i} L \bigg],
\end{equation}
where the maximum is over incentive-compatible contracts satisfying
participation \eqref{ir}.

Let $\pi^{NR}(\cdot;\phihat)$ be the project value in the baseline
model without resets. Then $\pi^{NR}(\cdot;\phi)$ is the first-best
value. We say that the optimal contract uses resets if no contract
without resets is optimal. The following lemma characterizes when this
is the case.

\begin{lemma}\label{lem:NecessarySuff}
  The optimal contract uses resets if and only if
  ${L < \pi^{NR}(x^R; \phihat) - \uv}$. The first best uses resets
  whenever the second best does, and only the first best uses them if
  and only if
  ${\pi^{NR}(x^R; \phihat) \leq L + \uv < \pi^{NR}(x^R; \phi)}$.
\end{lemma}

The intuition behind the lemma is straightforward. If
$L<\pi^{NR}(x^R;\phihat) - \underline v$, then whenever the principal
is on the brink of abandoning the project for the salvage value
$\underline v$, she prefers to pay the lump sum $L$ to jump to the
point $x^R$ where her continuation value is
$\pi^{NR}(x^R;\phihat)$. Conversely, if the cost of resetting is too
large relative to the gain in continuation value, it is better to
abandon the project. The comparison to the first best follows from the
mimicking argument in the proof of
Proposition~\ref{Prop:compfb}\ref{it:pcfb04}, which gives
$\pi^{NR}(\cdot; \phi) > \pi^{NR}(\cdot; \phihat)$ on $(\ux, \bar x)$
since $\phi < \phihat$.

We now examine how the principal should use resets, and how the
second-best policy differs from the first best. To that end, we assume
${L < \pi^{NR}(x^R; \phihat) - \uv}$. Because resets do not directly
affect the agent, the derivation of Section~\ref{Sec:LocalIC} applies
with one adjustment: \eqref{W-SDE} holds between resets, and the
agent's continuation utility is continuous at the reset times, which
are stopping times of progress, as martingales in a Brownian
filtration are continuous.  Lemma~\ref{Lem:IcNoS} and the recursive
formulation are then unchanged, and the principal's value takes the
form \eqref{fdefn} with a project value
$\pi$.\footnote{\label{fn:resets-verification}The verification
  argument in the proof of Proposition~\ref{Prop:MartingaleLevy}
  extends readily to resets, because the project still stops in finite
  time under any deviation and the sensitivity $\beta$ under the
  optimal contract is bounded as a continuous function of
  $x \in [\ux^R, \bar x]$.}

At a reset the principal's continuation value is $\pi(x^R) - L$. She
thus faces the no-reset problem with the salvage value
$\tilde v \ceq \pi(x^R) - L$ in place of $\uv$, where $\tilde v > \uv$
because the optimal no-reset contract remains feasible and
${L < \pi^{NR}(x^R; \phihat) - \uv}$. Hence the project is never
terminated but reset at a threshold $\ux^R$, and $\pi$ solves the
ODE~\eqref{eq:ODE} on the continuation region $[\ux^R, \bar x]$
subject to
\begin{equation}\label{eq:resetbc}
  \pi(\bar x) = b, \qquad \pi(\ux^R) = \pi(x^R) - L, \qquad
  \pi_x(\ux^R) = 0.
\end{equation}

The reset threshold $\ux^R$ lies above the termination threshold $\ux$
of the baseline model: the principal gives up on an approach sooner
than she would on the whole project. This is because the fallback is
now a fresh approach, which in expectation is worth more than the
salvage value from abandonment. The opportunity cost of continuing
with the current approach is thus higher, and the option to fall back
is exercised earlier.

The first-best problem is the ODE~\eqref{eq:ODE} subject to
\eqref{eq:resetbc} with $\phi$ in place of $\phihat$. We now compare
the two solutions.

\begin{proposition}\label{Prop:reset}
  Let $a$ and $\ux^R$ be the effort policy and reset threshold under
  the optimal contract, and let $a^{FB}$ and $\ux^{R,FB}$ be their
  first-best counterparts. Suppose $L < b - \uv$, so that both
  policies use resets for restart points close enough to the target.
\begin{enumerate}[label=(\roman*)]
\item \label{it:r03} Moral hazard makes the principal reset more
  aggressively when the restart point is close to the target: there
  exists $\varepsilon > 0$ such that $\ux^R > \ux^{R,FB}$ for all
  $x^R \in (\bar x - \varepsilon, \bar x)$.
\item \label{it:r04} Moral hazard can make the principal reset more
  conservatively: there exists $\varepsilon > 0$ such that if
  $\uv < \varepsilon$, then there are restart points $x^R$ with
  $0 < \pi^{NR}(x^R; \phihat) - \uv - L < \varepsilon$, and
  $\ux^R < \ux^{R,FB}$ for every such restart point.
\item \label{it:r06} If moral hazard makes the principal reset more
  aggressively, it lowers effort everywhere: if
  $\ux^R > \ux^{R,FB}$, then $a(x) < a^{FB}(x)$ for all
  $x \in (\ux^R, \bar x]$. If moral hazard makes the principal reset
  more conservatively, it raises effort at low progress levels, and
  the two effort policies cross at most once: if
  $\ux^R < \ux^{R,FB}$, there exists $x^* \in (\ux^{R,FB}, \bar x]$
  such that $a(x) > a^{FB}(x)$ for all $x \in (\ux^{R,FB}, x^*)$ and
  $a(x) < a^{FB}(x)$ for all $x \in (x^*, \bar x]$.
\end{enumerate}
\end{proposition}

To understand parts (i) and (ii), note that two forces determine how
moral hazard affects the reset threshold. On one hand, a reset
substitutes for effort and allows the principal to economize on
variable costs, which is more valuable under moral hazard as the
effective cost of effort is then higher. On the other hand, the value
of a reset also depends on how much progress is needed to reach the
target with the new approach, which reduces the value of a reset under
moral hazard because of the higher cost of effort. If the restart
point is sufficiently close to the target, the first force
dominates. This is part (i), which shows that a large number of
sufficiently promising approaches has the principal reaching for a new
one inefficiently often. If instead the work remaining after a reset
makes a fresh approach barely worth using under moral hazard, the
second force dominates. This is part (ii). A low salvage value puts
the principal's termination threshold, absent resets, far from the
target, and a fresh approach worth barely more than the salvage value
puts her reset threshold close to the termination threshold, hence
also far from the target. Figure~\ref{fig:Simulation} illustrates the
two cases.
\begin{figure}[t]
\centering
    \includegraphics[width=.48\linewidth]{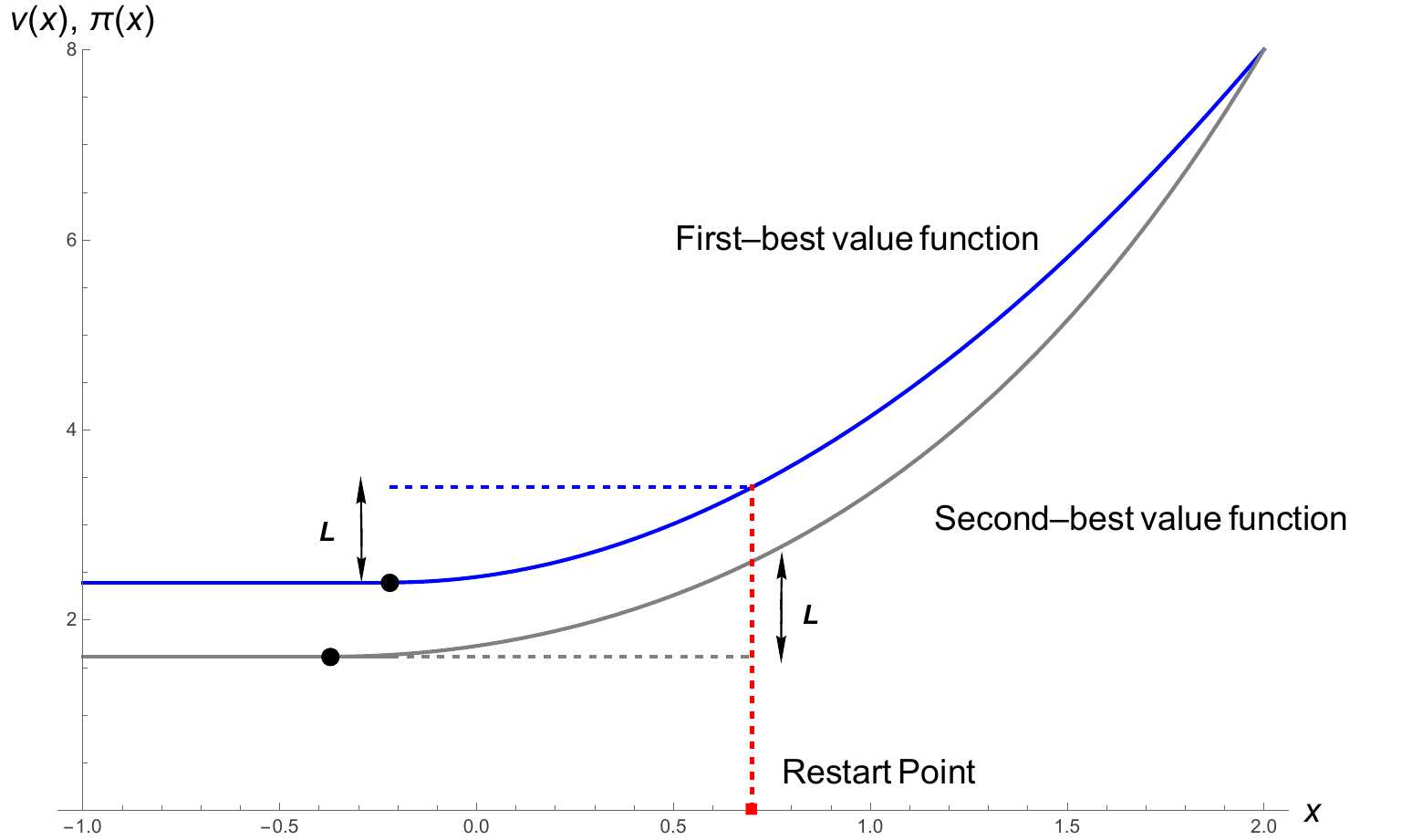}\quad
    \includegraphics[width=.48\linewidth]{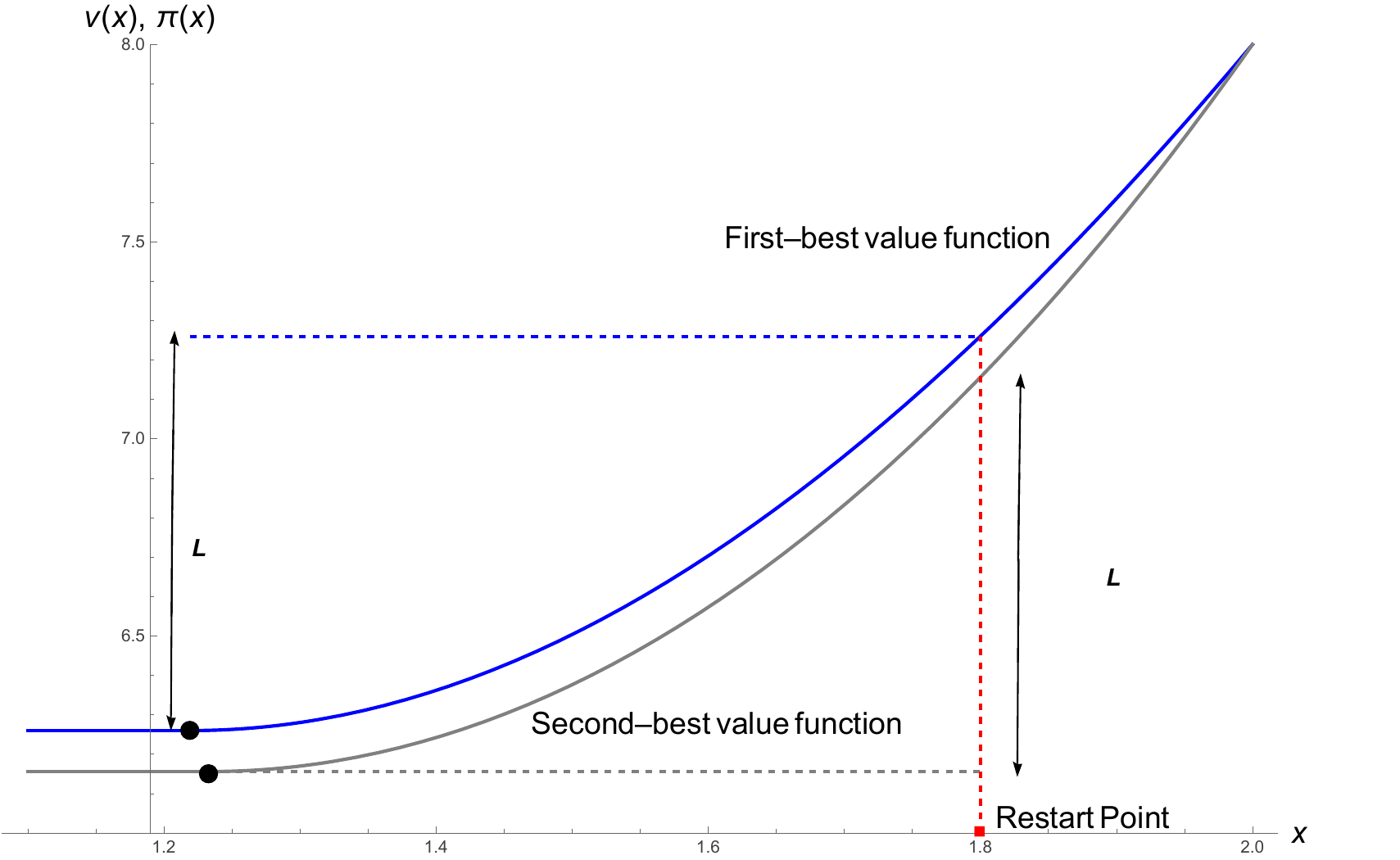}
    \caption{Restart point: $x^R = 0.7$ (left panel), $x^R = 1.8$
      (right panel). Other parameter values:
      $(b,L,\bar x,\sigma,\underline v,r,\eta,\phi
      )=(8,1,2,2,1,2,1/2,1)$.}\label{fig:Simulation}
\end{figure}

The effort comparisons in part (iii) follow from the ranking of reset
thresholds. Effort vanishes at the reset threshold, so whichever
policy resets more aggressively implements lower effort just above its
own threshold. In the case of part (i) this is the second best, and
moral hazard lowers effort everywhere, as in the baseline model. In
the case of part (ii) it is the first best, so moral hazard raises
effort at low progress levels, where the second best is still working
on an approach the first best has abandoned. The two effort policies
cross at most once, so the baseline ranking is restored at higher
progress levels.

\section{Risk of Breakdowns}\label{Sec:DisModel}

Many innovation projects, particularly in the life sciences, face
discrete, high-impact, and often irreversible setbacks. Clinical
trials may be terminated after adverse events, and early-stage biotech
programs may collapse upon a single scientific failure. These risks
are partly under the agent's control, as when he accelerates timelines
or pursues more aggressive protocols. We extend the baseline model to
allow for such risk taking.

In addition to choosing effort, the agent now makes a hidden choice
between a risky mode, $Q_t = 1$, and a safe mode, $Q_t = 0$. The risky
mode speeds up progress but raises the risk of a breakdown. In the
absence of a breakdown, progress evolves according to
\begin{equation}\label{eq:progress-risk}
  dX_t = (A_t + g Q_t)\, dt + \sigma\, dZ_t,
\end{equation}
where $g > 0$. Breakdowns arrive at rate $\lambda_0 > 0$ in the safe
mode and at the increased rate $\lambda_0 + \lambda_r$ in the risky
mode, where $\lambda_r > 0$. They are counted by a process $N$ with
intensity $\lambda_0 + \lambda_r Q_t$. The first breakdown terminates
the project at the salvage value $\uv$.

In this setting, a contract $(A, C, Q, \tau)$ also includes the
recommended risk mode $Q$ taking values in $\{0,1\}$, and all
components are progressively measurable with respect to progress and
breakdowns. Let $\zeta$ denote the first breakdown time. We adopt the
convention whereby the stopping time $\tau$ captures termination by a
breakdown in addition to completion and termination by choice. It is
thus required to satisfy $\tau \leq \zeta$.

As before, we focus on incentive-compatible contracts, under which the
agent obeys $A$ and $Q$ and consumes the payment $C$.  His time-$t$
continuation utility $W_t$ is thus as defined in
Section~\ref{Sec:LocalIC}, with the expectation taken under $(A,
Q)$. The following lemma gives necessary conditions for incentive
compatibility analogously to Lemmas~\ref{Lem:W} and \ref{Lem:IcNoS}.

\begin{lemma}\label{Lem:IcRisk}
  Let $(A, C, Q, \tau)$ be incentive compatible. The agent's
  continuation utility $W$ satisfies
  \begin{multline}\label{eq:W-risk}
    dW_t = \left( r W_t - u(A_t,C_t) \right) dt
    + \beta_t (-\eta r W_t) (dX_t - (A_t + g Q_t ) dt) \\
    + \psi_t (-\eta r W_{t-}) (dN_t - \left( \lambda_0 + \lambda_r Q_t \right) dt)
  \end{multline}
  for some progressively measurable processes $\beta$ and
  $\psi$. Moreover, $r W_t = u(A_t, C_t)$, and thus $W$ is a
  martingale. Consumption and effort are given by \eqref{cons} and
  \eqref{IC}, and the risk mode satisfies
  \begin{equation}\label{eq:IC-risk}
    Q_t\, (g \beta_t + \lambda_r \psi_t)
    = \max\{ g \beta_t + \lambda_r \psi_t,\, 0 \}.
  \end{equation}
\end{lemma}

The representation \eqref{eq:W-risk} follows from the martingale
representation theorem for the filtration generated by $Z$ and
$N$. The sensitivity of the agent's compensation to progress is
captured by $\beta_t$. Its sensitivity to a breakdown is measured by
$\psi_t$, with $\psi_t < 0$ a penalty. Effort and consumption are as
in Lemma~\ref{Lem:IcNoS}, because the risk mode is separable from
effort in the agent's problem. Finally, the risky mode raises the
drift by $g$ and the breakdown intensity by $\lambda_r$, so the agent
prefers it if and only if $g \beta_t + \lambda_r \psi_t \geq 0$.

The risk choice ties the two sensitivities together. Inducing effort
requires $\beta_t > 0$, which by itself makes the risky mode
attractive, as faster progress raises the agent's pay. To deter risk
taking, the principal must therefore penalize breakdowns by
$\psi_t \leq -g \beta_t / \lambda_r$, and the more effort she asks
for, the harsher the required penalty.

The conditions of Lemma~\ref{Lem:IcRisk} are necessary for incentive
compatibility. The next result shows that they are also sufficient
when the sensitivities are bounded and the contract is constant after
$\tau$. We allow $g = \lambda_0 = \lambda_r = 0$ to cover the baseline
model.\footnote{Proposition~\ref{Prop:BaseSuff} then follows: with
  $g = \lambda_0 = \lambda_r = 0$, breakdowns never occur, the risk
  mode is immaterial, and Lemma~\ref{Lem:IcRisk} reduces to
  Lemmas~\ref{Lem:W} and \ref{Lem:IcNoS}, whose conditions the
  contract in Proposition~\ref{Prop:BaseSuff} satisfies by
  construction. Further, the contract is constant after $\tau$, and
  its sensitivity $\beta = \phi A = \phi \phihat^{-1} \pi_x(X_t)$ is
  bounded, as $\pi_x$ is continuous on $[\ux, \bar x]$. (The result
  also extends to the model with resets; see
  footnote~\ref{fn:resets-verification}.)}

\begin{proposition}\label{Prop:MartingaleLevy}
  Let $g, \lambda_0, \lambda_r \geq 0$. Suppose $(A, C, Q, \tau)$ is a
  contract with $A_t = Q_t = 0$ and $C_t$ constant for $t > \tau$, and
  $W$ is a process such that the conditions of
  Lemma~\ref{Lem:IcRisk} hold with $\beta$ and $\psi$ bounded and
  $\psi \leq 0$. Then $(A, C, Q, \tau)$ is incentive compatible, and
  $W$ is the agent's continuation utility.
\end{proposition}

Our novel proof verifies global incentive compatibility against every
deviation in effort, risk, and consumption satisfying the no-Ponzi
condition. As is standard, one-shot optimality makes the agent's
payoff from deviating up to time $t$ and obeying thereafter a local
supermartingale in $t$. Rather than restricting deviations by
integrability conditions to pass to expectations, we localize: stopped
at times when the value of returning to obedience is bounded, the
payoff is a true supermartingale. To let these times go to infinity,
we bound the value of returning to obedience using only the almost
sure no-Ponzi condition and the boundedness of $\beta$ and $\psi$ (the
sign restriction on $\psi$ simplifies the argument but is not
essential).

The principal's problem is \eqref{second-best}, with the maximum now
over incentive-compatible contracts $(A, C, Q, \tau)$ and the
expectation taken under $(A, Q)$. As in Section~\ref{Sec:Recursive},
we impose only the conditions of Lemma~\ref{Lem:IcRisk}; this is
without loss by Proposition~\ref{Prop:MartingaleLevy} provided that
the resulting contract has bounded sensitivities. We also guess that
the value function takes the form \eqref{fdefn}. The HJB equation for
the project value $\pi$ is then
\begin{multline}\label{hjb:risk}
  r \pi(x) = \max_{a, q, \psi} \Big\{ -\frac{\phihat}{2} a^2
  + (a + g q)\, \pi_x(x) + \frac{\sigma^2}{2} \pi_{xx}(x) \\
  + (\lambda_0 + \lambda_r q) \Big( \uv - \pi(x) + \psi
  + \frac{1}{\eta r} \ln(1 - \eta r \psi) \Big) \Big\},
\end{multline}
where the maximum is subject to \eqref{eq:IC-risk}. The last term
captures the risk of a breakdown: if one arrives, the principal
receives $\uv$
in place of $\pi(x)$ and collects the drop
$\frac{1}{\eta r} \ln(1 - \eta r \psi) = CE(W_{t-}) - CE(W_t)$ in the
agent's certainty equivalent. The term $(\lambda_0 + \lambda_r q)\psi$
is the flow compensation for the penalty, paid while no breakdown
occurs. In expectation, a penalty thus costs its risk premium
$-\psi - \frac{1}{\eta r} \ln(1 - \eta r \psi) > 0$ per unit of
intensity.

Solving the HJB equation \eqref{hjb:risk} subject to value matching
\eqref{Val} and smooth pasting \eqref{past} yields a project value
$\pi$ and a termination threshold $\ux$. Define a contract
$(A,C,Q,\tau)$ from the solution as in Section~\ref{Sec:Recursive}:
Let $\tau$ be the first time $X_t \notin (\ux, \bar x)$, or the first
breakdown if earlier. For $t \leq \tau$, let $(A_t, Q_t, \psi_t)$ be
the maximizer in \eqref{hjb:risk} at $x = X_t$ and
$\beta_t = \phi A_t$, and determine $C_t$ and $W_t$ via
\eqref{eq:W-risk} and \eqref{cons}. For $t > \tau$, set
$A_t = Q_t = 0$ and $C_t = r\, CE(W_\tau)$.

The contract so obtained is incentive compatible by
Proposition~\ref{Prop:MartingaleLevy}: it is constant after $\tau$ and
satisfies the conditions of Lemma~\ref{Lem:IcRisk} by
construction. Moreover, the sensitivities $\beta = \phi A$ and
$\psi \leq 0$ are bounded, as they are continuous functions of
$\pi_x(X_t)$ in each risk mode and $\pi_x$ is continuous on
$[\ux, \bar x]$. We have thus found an optimal contract. The following
proposition summarizes its properties.

\begin{proposition}\label{Prop:OptimalContract}
  Let $\pi$, $a$, and $\ux$ be the project value, effort policy, and
  termination threshold under the optimal contract, and let $\psi$ be
  the breakdown penalty. Suppose that $\lambda_r$ or $\sigma$ is large
  enough.\footnote{The proof gives explicit sufficient conditions,
    \eqref{eq:sigma_cond} and \eqref{eq:suff2}.} In the continuation
  region $[\ux, \bar x]$, the following properties hold:
  \begin{enumerate}[label=(\roman*)]
  \item\label{it:brp01} Value is increasing and convex in progress:
    $\pi_x > 0$ on $(\ux, \bar x]$ and $\pi_{xx} > 0$.
  \item\label{it:brp02} The risky mode is optimal below some cutoff
    ${x_c^{SB} \in (\ux, \bar x]}$ and the safe mode above it.
  \item\label{it:brp05} The safe mode is optimal near completion for a
    high benefit: $x_c^{SB} < \bar x$ for $b$ large enough.
  \item\label{it:brp04} Effort is continuous and increasing except at
    $x_c^{SB}$, where it jumps down if $x_c^{SB} < \bar x$.
  \item\label{it:brp03} Breakdowns are penalized only in the safe mode,
    in proportion to effort: $\psi = 0$ below $x_c^{SB}$ and
    $\psi = -\phi g a/\lambda_r$ above it.
  \end{enumerate}
\end{proposition}

The principal encourages risk taking when little is at stake and then
shuts it down once enough progress has been made, possibly only at
completion (parts (ii) and (iii)). This pattern also arises in the
first best, where effort and risk choices are observable.\footnote{See
  the Supplemental Appendix for the analysis of the first-best
  benchmark.} In the risky mode below $x^{SB}_c$, the agent is not
penalized for breakdowns, but he is not rewarded for them either. The
pay-performance sensitivity on its own already motivates risk taking,
so a bonus upon breakdown would add compensation risk without relaxing
any constraints. In the safe mode above $x^{SB}_c$, the agent is
penalized for breakdowns to deter risk taking. As the penalty is
costly, it is set to be as small as possible (part (v)). Penalties
therefore grow with effort, and hence also with progress (part (iv)).

\begin{figure}[t]
\centering
\begin{minipage}[c]{.48\linewidth}
\centering
\includegraphics[width=\linewidth]{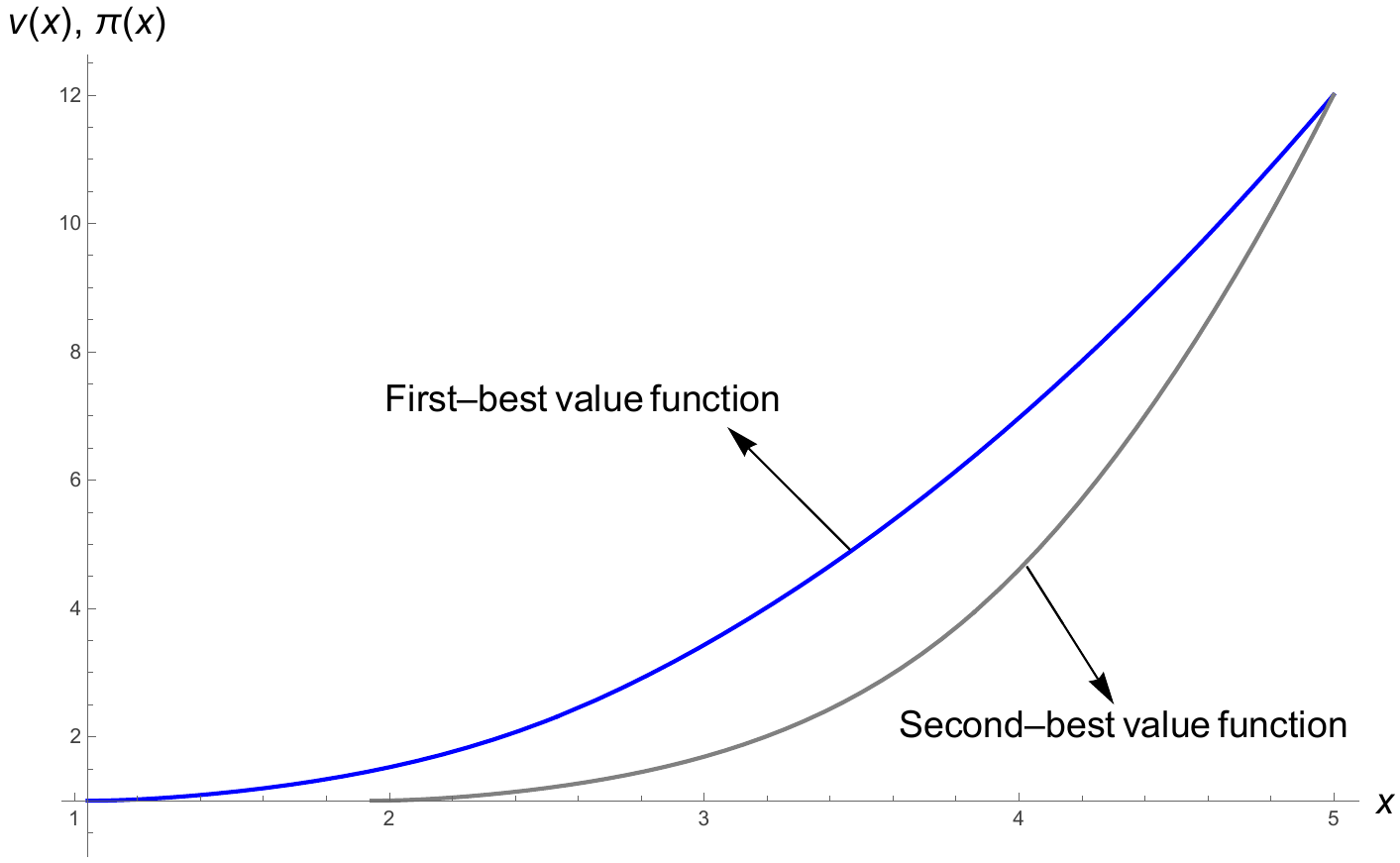}
\end{minipage}\quad
\begin{minipage}[c]{.48\linewidth}
\centering
\includegraphics[width=\linewidth]{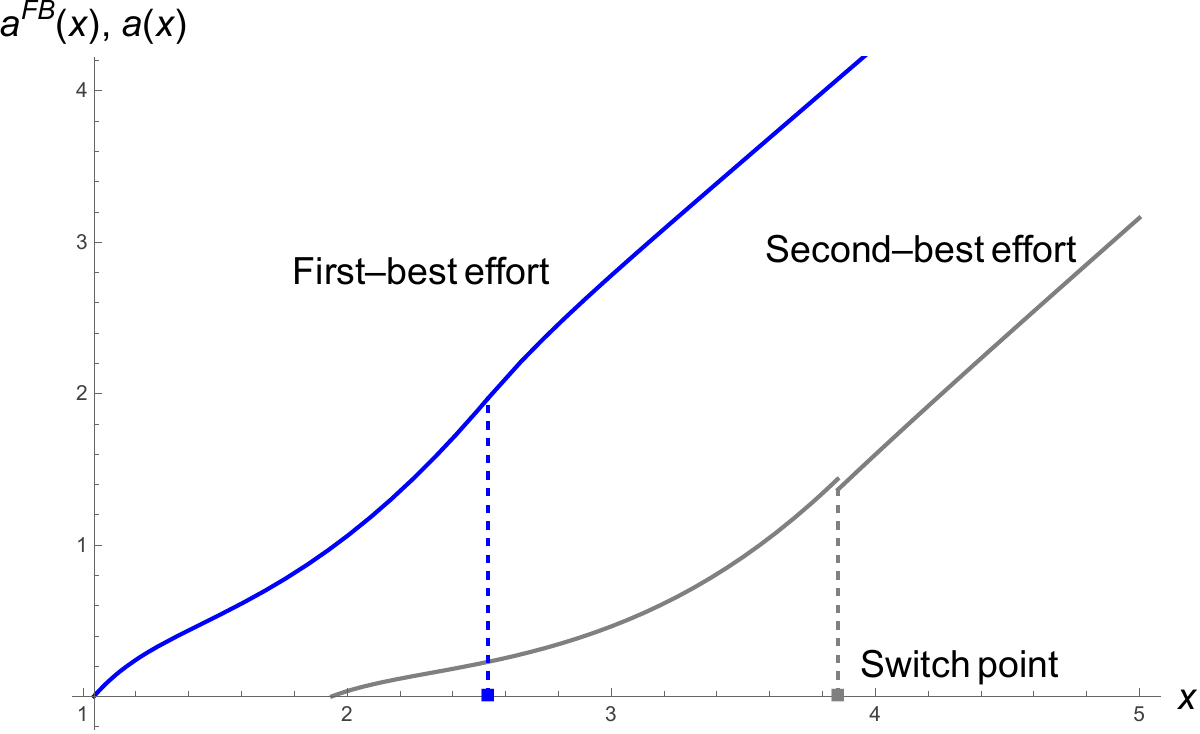}
\end{minipage}
\caption{Comparing the first best and second best under endogenous
  risk taking. Parameter values:
  $(b, \bar x, \sigma, \underline v, r, \eta, \phi, \lambda_0,
  \lambda_r, g) = (12, 5, 1, 1, 1, 2, 1, 1/2, 3,
  2)$. \label{fig:risk}}
\end{figure}

Interestingly, the agent's two hidden actions lead the principal to
implement effort that is discontinuous at the cutoff $x_c^{SB}$
between the risky and safe modes (part (iv)). At the cutoff, the
marginal benefit of effort is simply the marginal value of
progress. However, whereas in the risky mode effort $a$ requires only
the pay-performance sensitivity $\beta = \phi a$, in the safe mode it
also requires the penalty $\psi = -\phi g a/\lambda_r$ to deter risk
taking, which costs the principal its risk premium. The marginal cost
of effort is thus higher in the safe mode, leading to the downward
jump in effort at the cutoff. Effort may even peak in the interior of
the continuation region rather than at completion. In the
first best, by contrast, the risk mode does not affect the cost of
effort, and effort is continuous everywhere. \autoref{fig:risk}
illustrates.

 \section{Conclusion \label{Sec:Conc}}

 This paper studies optimal contracting for projects with gradual
 progress and lumpy success. Our modeling framework provides a
 tractable platform for studying dynamic incentives for complex
 projects. The results highlight several novel features of optimal
 project management. The presence of moral hazard leads to slower
 progress and lowers the probability of success. It makes the
 principal less failure tolerant and less ambitious than under
 observable actions. These features imply that in general a
 profit-maximizing principal will choose to pursue different projects
 than a social planner.\footnote{E.g., we show in the Supplemental
   Appendix that the principal may prefer a project with lower
   volatility even when the social planner selects the
   higher-volatility one for its increased option value.}

 If multiple approaches are available, the principal substitutes a
 fixed resetting cost for the agency costs of recovering from
 setbacks, but moral hazard can lead to more frequent or less frequent
 resets, depending on the expected amount of residual effort. Under
 endogenous breakdown risk, the principal uses discontinuous effort
 profiles and nonmonotonic incentive schemes to deter risk taking or
 to speed up the project's progress.
 
 Future research could extend our analysis in several directions. One
 possibility is to explore richer informational frictions, such as
 noisy or delayed feedback on project progress, common in R\&D, drug
 development, and engineering design. Another is to introduce
 uncertainty over project quality, with the productivity of effort
 initially unknown and learned by both parties from realized
 progress. In such a setting, effort serves a dual role: it brings the
 project closer to completion and generates information about the
 technology's viability.

\newpage
\section*{Appendix}
\setcounter{equation}{0}
\renewcommand{\theequation}{A.\arabic{equation}}
\renewcommand{\theHequation}{A.\arabic{equation}}
\setcounter{lemma}{0}
\renewcommand{\thelemma}{A.\arabic{lemma}}
\renewcommand{\theHlemma}{A.\arabic{lemma}}

\subsection*{Savings and Continuation Values}
Lemma~\ref{lem:AWE} shows that savings simply scale the agent's
continuation utility, because additional wealth is optimally consumed
as an annuity (cf.\ \citealp{he2011model}, Lemma 3). We state it for
all versions of the model at once. Let $(\mathcal F_t)$ denote the
filtration generated by the public history, that is, by progress and,
where present, breakdowns.

\begin{lemma}\label{lem:AWE}
  Fix $t \geq 0$ and a payment process $P$ adapted to
  $(\mathcal F_t)$. Let $\hat \Xi$ denote an action process (effort,
  or effort and risk) with effort component $\hat A$, and let
  \begin{equation*}
    W_t(S; P) \ceq \sup_{\hat \Xi, \hat C}
    \E^{\hat \Xi}\!\left[\int_t^{\infty} e^{-r(s-t)}\, u(\hat A_s, \hat C_s)\,
      ds \;\Big|\; \mathcal{F}_t\right]
  \end{equation*}
  be the agent's continuation utility at time $t$ with savings
  $S \in \R$, where the supremum is over $(\mathcal F_t)$-progressively
  measurable $(\hat \Xi, \hat C)$ satisfying
  $dS_s = (r S_s - \hat C_s)\, ds + dP_s$ with $S_t = S$ and the
  no-Ponzi condition. Then, almost surely,
  $W_t(S; P) = e^{-\eta r S}\, W_t(0; P)$ for all $S$. In particular,
  $W_t(S; 0) = CE^{-1}(S)$: the outside option is worth $S$ to an
  agent with savings $S$.
\end{lemma}

\begin{proof}
  Let $(\hat \Xi, \hat C)$ be feasible from savings $0$, with savings
  path $S^0$. Then $(\hat \Xi, \hat C + rS)$ is feasible from savings
  $S$, with savings path $S^0 + S$, since
  \begin{equation*}
    d(S^0_s + S) = \big( r (S^0_s + S) - (\hat C_s + rS) \big)\, ds
    + dP_s
  \end{equation*}
  and $e^{-rs}(S^0_s + S) \to 0$ if and only if
  $e^{-rs} S^0_s \to 0$. By the same argument in reverse,
  $\hat C \mapsto \hat C + rS$ is a bijection between the consumption
  processes feasible from $0$ and from $S$ for each $\hat \Xi$. The
  law of the public history is the same under both, and by
  \eqref{carau},
  \begin{equation*}
    u(\hat A_s, \hat C_s + rS) = e^{-\eta r S}\, u(\hat A_s, \hat C_s),
  \end{equation*}
  so the payoff of every policy is multiplied by $e^{-\eta r S}$.
  Taking suprema gives the first claim. For $P = 0$, consuming
  nothing from zero savings is optimal and yields
  $W_t(0; 0) = -1/(\eta r) = CE^{-1}(0)$. Therefore,
  $W_t(S; 0) = e^{-\eta r S} CE^{-1}(0) = CE^{-1}(S)$ by
  \eqref{eq:CE}.
\end{proof}

\subsection*{Proof of Proposition \ref{Prop:FirstBest}}

\begin{lemma}\label{lem-convex}
  If $\hat x \in [\underline x^{FB}, \bar x]$ satisfies $v_x(\hat x) = 0$
  and $v_x(x) \geq 0$ $\forall x \in [\underline x^{FB},\hat x]$, then
  $v_{xx}(\hat x) >0$.
\end{lemma}
\begin{proof}
  We have $ v(\hat x) \geq v(\underline x^{FB}) = \underline v >0$, where
  the first inequality follows since $v_x \geq 0$ on
  $[\underline x^{FB},\hat x]$, and the equality is by value matching at
  $\underline x^{FB}$. Evaluating the HJB equation \eqref{hjb:planner}
  at $\hat x$ gives $v_{xx} (\hat x) = 2\sigma^{-2}r v(\hat x) >0$, since
  $a^{FB}(\hat x) =\phi^{-1} v_x (\hat x)=0$ by the FOC \eqref{FB-FOC}.
\end{proof}

\begin{proof}[Proof of part \ref{it:d02}]
  Smooth pasting implies that $v_x(\underline x^{FB})=0$. Taking
  $\hat x= \underline x^{FB}$ in Lemma~\ref{lem-convex} then shows
  that $v_x>0$ on $(\underline x^{FB},\underline x^{FB} +\epsilon)$
  for some $\epsilon >0$. Suppose toward a contradiction that the set
  $O\ceq \{x \in [\underline x^{FB} +\epsilon, \bar x]: v_x(x)\leq
  0\}$ is nonempty. Let $x^* \ceq \inf O$. By continuity,
  $v_x(x^*) = 0$. Moreover, $v_x \geq 0$ on $[\underline x^{FB},x^*]$
  by definition of $x^*$. Taking $\hat x = x^*$ in
  Lemma~\ref{lem-convex} then implies that $v_x$ is increasing at
  $x^*$, which contradicts $x^*$ being the greatest lower bound of
  $O$. We conclude that $O$ is empty. Hence, $v_x>0$ on
  $(\underline x^{FB},\bar x]$.
\end{proof}
\begin{proof}[Proof of part \ref{it:d03}]
  Lemma~\ref{lem-convex} and smooth pasting imply that
  $v_{xx}(\underline x^{FB})>0$. Suppose toward a contradiction that the
  set $O\ceq \{x \in [\underline x^{FB}, \bar x]: v_{xx}(x)\leq 0\}$ is
  nonempty. Let $x^* \ceq \inf O$. By continuity, $v_{xx}(x^*) = 0$ and
  thus $x^*>\underline x^{FB}$. Differentiating the HJB equation
  \eqref{hjb:planner} with \eqref{FB-FOC} substituted in gives, for
  all $x \in [\underline x^{FB}, \bar x]$,
  \begin{equation}\label{eq:Envelope}
    r v_x(x) = \phi^{-1} v_x(x) v_{xx}(x) + \tfrac{1}{2} \sigma^2
    v_{xxx}(x).
  \end{equation}
  At $x^*$, this gives $v_{xxx}(x^*) > 0$ by part (i), so $v_{xx}$ is
  increasing at $x^*$, contradicting $x^*$ being the greatest lower
  bound of $O$. We conclude that $O$ is empty.
\end{proof}
\begin{proof}[Proof of part \ref{it:d04}]
  The result is immediate from the FOC \eqref{FB-FOC} and part (ii).
\end{proof}
\begin{proof}[Proof of part \ref{it:d05}]
  By \eqref{FB-FOC}, $a^{FB}(X_t) = \phi^{-1} v_x(X_t)$, so It\^o's
  formula gives, for $t < \tau$,
  \begin{equation*}
    d a^{FB}(X_t) = \phi^{-1} \big[ v_{xx}(X_t)\, dX_t + \tfrac{1}{2}
      \sigma^2 v_{xxx}(X_t)\, dt \big] = r\, a^{FB}(X_t)\, dt +
    \phi^{-1} \sigma v_{xx}(X_t)\, dZ_t,
  \end{equation*}
  where the last equality uses \eqref{progress} and
  \eqref{eq:Envelope}. Thus
  $d \big(e^{-rt} a^{FB}(X_t)\big) = e^{-rt} \phi^{-1} \sigma
  v_{xx}(X_t)\, dZ_t$, and since $v_{xx}$ is bounded on
  $[\underline x^{FB}, \bar x]$, the stopped process is a martingale.
\end{proof}

\subsection*{Proof of Proposition \ref{Prop:compfb}}
\begin{proof}[Proof of part \ref{it:pcfb01}]
  Fix parameters such that $b_1 > b_2$. The HJB equation
  \eqref{hjb:planner} with the FOC \eqref{FB-FOC} substituted in is
  autonomous in $x$: the right-hand side depends on $v$, $v_x$, and
  $v_{xx}$ but not on $x$ directly. Let $g$ denote the unique solution
  to
\[
  r g = \frac{1}{2\phi}(g')^2 + \frac{1}{2}\sigma^2 g''
\]
with initial conditions $g(0) = \underline v$ and $g'(0) = 0$. By
Proposition~\ref{Prop:FirstBest}, $g$ is increasing and strictly
convex on $(0,\infty)$, and $g' > 0$ is increasing.

By autonomy, each first-best value function is a horizontal translate
of $g$: For $i=1,2$, we have $v_i(x) = g(x - \underline x^{FB}_i)$ for
$x \in [\underline x^{FB}_i, \bar x]$. Value matching 
$v_i(\bar x) = b_i$ requires
\[
  g(\bar x - \underline x^{FB}_i) = b_i.
\]
Since $g$ is increasing and $b_1 > b_2$, it follows that
$\bar x - \underline x^{FB}_1 > \bar x - \underline x^{FB}_2$, i.e.,
$\underline x^{FB}_1 < \underline x^{FB}_2$.

For the effort ranking, note that
$a^{FB}_i(x) = \phi^{-1} v_{i,x}(x) = \phi^{-1} g'(x - \underline
x^{FB}_i)$.  Since $g'$ is increasing and
$x - \underline x^{FB}_1 > x - \underline x^{FB}_2$, we have
$a^{FB}_1(x) > a^{FB}_2(x)$ for all
$x \in (\underline x^{FB}_2, \bar x]$. The comparison extends to
$(\underline x^{FB}_1, \bar x]$ since project 2 is terminated in this
region.
\end{proof}

\begin{proof}[Proof of part~\ref{it:pcfb02}]
  Fix parameters such that $r_1 > r_2$. A mimicking argument implies
  that $v_1(x) \leq v_2(x)$ for all $x$ and
  $\underline{x}^{FB}_1 \geq \underline{x}^{FB}_2$.\footnote{Realized
    payoffs are not ranked path by path, as a higher discount rate
    also reduces the present value of the effort costs. However,
    because the expected continuation payoff under project~1's policy
    is bounded below by the salvage value $\underline v > 0$ at every
    history, a decrease in the discount rate raises the expected
    payoff at every history, and in particular at time zero.} We first show
  $\underline{x}^{FB}_1 > \underline{x}^{FB}_2$. If
  $\underline{x}^{FB}_1 = \underline{x}^{FB}_2 \eqqcolon \underline{x}$, then
  $v_{i,xx}(\underline{x}) = 2r_i \underline{v}/\sigma^2$ by the
  HJB. Since $r_1 > r_2$, we have
  $v_{1,xx}(\underline{x}) > v_{2,xx}(\underline{x})$, so $v_1 > v_2$
  in a right neighborhood, contradicting $v_1 \leq v_2$.

  Let $D \ceq v_2 - v_1$ on $[\ux^{FB}_2, \bar x]$, with $v_1 \ceq \uv$
  below $\ux^{FB}_1$. Then $D$ is continuously differentiable,
  $D(\ux^{FB}_1) > 0 = D(\bar x)$, and $D_x = v_{2,x} > 0$ on
  $(\ux^{FB}_2, \ux^{FB}_1]$. By \eqref{FB-FOC},
  $a^{FB}_1 - a^{FB}_2 = -\phi^{-1}D_x$, so it suffices to show that $D_x$ has
  a unique zero $x^\dagger_r$ in $(\ux^{FB}_1, \bar x)$, with
  $D_x < 0$ above it. A zero exists as $D(\ux^{FB}_1) > D(\bar x)$;
  let $x^\dagger_r$ be the first. Let $T \ceq r_1 v_1 - r_2 v_2$. At
  a zero of $D_x$, subtracting the HJB equations gives
  $\frac{\sigma^2}{2}D_{xx} = -T$, and if also $D_{xx} = 0$,
  subtracting \eqref{eq:Envelope} for the two projects gives
  $\frac{\sigma^2}{2}D_{xxx} = -(r_1 - r_2)v_{1,x} < 0$. Since
  $D_x > 0$ to the left of $x^\dagger_r$, $D_{xx}(x^\dagger_r) \leq 0$,
  and $D_{xx}(x^\dagger_r) = 0$ is impossible, as $D_x$ would then
  have a strict local maximum at $x^\dagger_r$. Hence
  $T(x^\dagger_r) > 0$. Moreover,
  $T_x = (r_1 - r_2)v_{1,x} - r_2 D_x > 0$ wherever $D_x \leq 0$. If
  $D_x$ had a second zero, then at the first such point $x^\ddagger$
  we would have $T(x^\ddagger) > T(x^\dagger_r) > 0$, so
  $D_{xx}(x^\ddagger) < 0$, whereas $D_x$ reaching zero from below
  requires $D_{xx}(x^\ddagger) \geq 0$.
\end{proof}

\begin{proof}[Proof of part \ref{it:pcfb03}]
    Fix parameters such that $\sigma_1 > \sigma_2$. First,
  $v_1 \geq v_2$: implement project 2's effort policy $a^{FB}_2$ in
  project 1 and stop when progress leaves $[\ux^{FB}_2, \bar x]$. By
  project 2's HJB equation and the convexity of $v_2$
  (Proposition~\ref{Prop:FirstBest}\ref{it:d03}),
  $a^{FB}_2 v_{2,x} + \frac{\sigma_1^2}{2}v_{2,xx} - r v_2 =
  \frac{\phi}{2}(a^{FB}_2)^2 + \frac{\sigma_1^2 - \sigma_2^2}{2}v_{2,xx}
  \geq \frac{\phi}{2}(a^{FB}_2)^2$, so by It\^o's formula $v_2(x)$ is
  at most the payoff of this policy in project 1, which is at most
  $v_1(x)$; the inequality is strict on $(\ux^{FB}_2, \bar x)$ since
  $v_{2,xx} > 0$. Hence $\ux^{FB}_1 \leq \ux^{FB}_2$ as well.

  We first show $\underline x^{FB}_1 < \underline x^{FB}_2$. Suppose
  toward a contradiction that
  $\underline x^{FB}_1 = \underline x^{FB}_2 \eqqcolon \underline
  x$. Then $v_i(\underline x) = \underline v$ and
  $v_{i,x}(\underline x) = 0$ for $i = 1,2$. Evaluating the HJB
  \eqref{hjb:planner} at $\underline x$ gives
  $v_{i,xx}(\underline x) = 2r\underline v/\sigma_i^2$. Since
  $\sigma_1 > \sigma_2$, we have
  $v_{1,xx}(\underline x) < v_{2,xx}(\underline x)$, so $v_2 > v_1$ in
  a right neighborhood of $\underline x$, contradicting
  $v_1 \geq v_2$.

  Given $\underline x^{FB}_1 < \underline x^{FB}_2$, the function
  $D \ceq v_1 - v_2$ satisfies $D(x) = 0$ for all
  $x \leq \underline x^{FB}_1$, $D(\underline x^{FB}_2) > 0$ (since
  $v_1(\underline x^{FB}_2) > \underline v = v_2(\underline
  x^{FB}_2)$), $D(x) \geq 0$ for all $x$, and $D(\bar x) = 0$. Note
  that $D_x = v_{1,x} - v_{2,x} = \phi(a^{FB}_1 - a^{FB}_2)$ by
  \eqref{FB-FOC}. Since $v_{2,x}(\underline x^{FB}_2) = 0$, we have
  $D_x(\underline x^{FB}_2) = v_{1,x}(\underline x^{FB}_2) > 0$, hence
  $a^{FB}_1(\underline x^{FB}_2) > a^{FB}_2(\underline x^{FB}_2) = 0$.

  We next show $D_x(\bar x) < 0$. Since $D \geq 0$ and
  $D(\bar x) = 0$, we have $D_x(\bar x) \leq 0$. If $D_x(\bar x) = 0$,
  then $v_{1,x}(\bar x) = v_{2,x}(\bar x)$ and the HJB equations at
  $\bar x$ give
  $\sigma_1^2 v_{1,xx}(\bar x) = \sigma_2^2 v_{2,xx}(\bar x)$. Since
  $\sigma_1 > \sigma_2$, this implies
  $D_{xx}(\bar x) = v_{1,xx}(\bar x) - v_{2,xx}(\bar x) < 0$. Combined
  with $D(\bar x) = D_x(\bar x) = 0$, this gives $D < 0$ in a left
  neighborhood of $\bar x$, contradicting $D \geq 0$. Hence
  $D_x(\bar x) < 0$, i.e., $a^{FB}_1(\bar x) < a^{FB}_2(\bar x)$.

  Define $x^{\dagger}$ as the infimum of
  $\{x > \underline x^{FB}_2 : D_x(x) = 0\}$ and $x^{\dagger\dagger}$
  as the supremum of $\{x < \bar x : D_x(x) = 0\}$. Since
  $D_x(\underline x^{FB}_2) > 0$ and $D_x(\bar x) < 0$, both
  thresholds exist and are interior. By construction, $D_x > 0$ on
  $(\underline x^{FB}_2, x^{\dagger})$ and $D_x < 0$ on
  $(x^{\dagger\dagger}, \bar x)$, which gives
  $a^{FB}_1(x) > a^{FB}_2(x)$ for
  $x \in (\underline x^{FB}_2, x^{\dagger})$ and
  $a^{FB}_1(x) < a^{FB}_2(x)$ for
  $x \in (x^{\dagger\dagger}, \bar x)$. For
  $x \in (\underline x^{FB}_1, \underline x^{FB}_2)$,
  $a^{FB}_1(x) > 0 = a^{FB}_2(x)$ since project 2 is terminated in
  this region.

\end{proof}

\begin{proof}[Proof of part \ref{it:pcfb04}]
  Fix parameters such that $\phi_1 > \phi_2$. A standard mimicking
  argument implies that $v_1(x) \leq v_2(x)$ for all $x$ and
  $\underline x^{FB}_1 \geq \underline x^{FB}_2$. Moreover, the value ranking
  is strict wherever project~1 continues: for $x \in (\underline x^{FB}_1, \bar x)$,
  implementing project~1's optimal effort and stopping policy for
  project~2 yields
  \[
    v_2(x) \;\geq\; v_1(x) + \E_x\left[\int_0^{\tau_1}
      e^{-rt}\,\frac{\phi_1-\phi_2}{2}\,
      \big(a^{FB}_1(X_t)\big)^2\,dt\right] \;>\; v_1(x),
  \]
  where $\tau_1$ denotes project~1's stopping time, and the second
  inequality holds because, started at $x \in (\underline x^{FB}_1, \bar x)$, the
  process spends positive expected discounted time in the region where
  $a^{FB}_1 > 0$ before $\tau_1$.

  We next show $\underline x^{FB}_1 > \underline x^{FB}_2$. Suppose toward a
  contradiction that
  $\underline x^{FB}_1 = \underline x^{FB}_2 \eqqcolon \underline x$. For
  $i = 1,2$, $v_i$ is infinitely differentiable on
  $[\underline x, \bar x]$ and satisfies the HJB equation
  \eqref{hjb:planner} with the FOC \eqref{FB-FOC} substituted in;
  moreover, $v_i(\underline x) = \underline v$ and
  $v_{i,x}(\underline x) = 0$ by value matching and smooth
  pasting. Evaluating the ODE, \eqref{eq:Envelope}, and the derivative
  of \eqref{eq:Envelope} at $\underline x$ gives
  \[
    v_{i,xx}(\underline x) = \frac{2r\underline v}{\sigma^2}, \qquad
    v_{i,xxx}(\underline x) = 0, \qquad
    v_{i,xxxx}(\underline x) =
    \frac{2}{\sigma^2}\left(r\,v_{i,xx}(\underline x)
      - \frac{v_{i,xx}(\underline x)^2}{\phi_i}\right).
  \]
  Thus $v_1$ and $v_2$ agree at $\underline x$ up to third order,
  while their fourth derivatives satisfy
  \[
    v_{2,xxxx}(\underline x) - v_{1,xxxx}(\underline x)
    = -\frac{2}{\sigma^2}
    \left(\frac{2r\underline v}{\sigma^2}\right)^2
    \left(\frac{1}{\phi_2}-\frac{1}{\phi_1}\right) < 0.
  \]
  By Taylor's theorem, $v_2(x) < v_1(x)$ for all $x$ in a right
  neighborhood of $\underline x$, contradicting the strict ranking
  established above. We conclude that
  $\underline x^{FB}_1 > \underline x^{FB}_2$.

  Given $\underline x^{FB}_1 > \underline x^{FB}_2$, at $\underline x^{FB}_1$ we
  have $a^{FB}_1(\underline x^{FB}_1) = 0 < a^{FB}_2(\underline x^{FB}_1)$. Suppose
  toward a contradiction that there exists a first crossing point
  $x^* > \underline x^{FB}_1$ where $a^{FB}_1(x^*) = a^{FB}_2(x^*) \eqqcolon
  a$. Since $a^{FB}_1$ approaches $a^{FB}_2$ from below,
  $a^{FB}_{1,x}(x^*) \geq a^{FB}_{2,x}(x^*)$. Equal efforts imply
  $v_{1,x}(x^*) = \phi_1 a > \phi_2 a = v_{2,x}(x^*)$. Subtracting the
  HJB equations at $x^*$,
  \[
    r\bigl(v_1(x^*) - v_2(x^*)\bigr) = \tfrac{1}{2}(\phi_1 - \phi_2)\,
    a^2 + \tfrac{1}{2}\sigma^2 \bigl(v_{1,xx}(x^*) -
    v_{2,xx}(x^*)\bigr).
  \]
  The left-hand side is nonpositive (since $v_1 \leq v_2$) and the
  first term on the right is strictly positive (since
  $\phi_1 > \phi_2$ and $a > 0$), so $v_{1,xx}(x^*) <
  v_{2,xx}(x^*)$. Hence
  $a^{FB}_{1,x}(x^*) = \phi_1^{-1} v_{1,xx}(x^*) < \phi_2^{-1}
  v_{2,xx}(x^*) = a^{FB}_{2,x}(x^*)$, contradicting
  $a^{FB}_{1,x}(x^*) \geq a^{FB}_{2,x}(x^*)$.
\end{proof}

\subsection*{Proof of Proposition \ref{Prop:SecondBest}}
Parts \ref{it:s01}, \ref{it:s02}, and \ref{it:s04} follow by
Proposition~\ref{Prop:FirstBest}, and part~\ref{it:s03} from
$a = \phihat^{-1} \pi_x$ and $\beta = \phi a$.

\subsection*{Proof of Proposition~\ref{Prop:comp}}

\begin{proof}[Proof of parts~\ref{it:pc01}--\ref{it:pc03}]
  The benefit $b$ enters the auxiliary problem \eqref{prpr} as it
  enters the first-best problem, and $\eta$ and $\phi$ enter only
  through $\phihat$, which is increasing in each, so the three parts
  are Proposition~\ref{Prop:compfb}\ref{it:pcfb01} and \ref{it:pcfb04}
  applied to the auxiliary problem.
\end{proof}

\begin{proof}[Proof of part~\ref{it:pc04}]
  Let $\ux(r, \phihat)$ denote the termination threshold of the
  auxiliary problem with discount rate $r$ and cost parameter
  $\phihat$, so that $\ux_i = \ux(r_i, \phihat_i)$ with
  $\phihat_1 > \phihat_2$. Holding $\phihat$ fixed, a higher discount
  rate raises the threshold by the argument in the proof of
  Proposition~\ref{Prop:compfb}\ref{it:pcfb02}; holding $r$ fixed, a
  higher cost parameter raises it by
  Proposition~\ref{Prop:compfb}\ref{it:pcfb04}. Hence
  $\ux_1 = \ux(r_1, \phihat_1) > \ux(r_1, \phihat_2) > \ux(r_2,
  \phihat_2) = \ux_2$.

  For effort, let $\theta_i \ceq \pi_i / \phihat_i$, so that
  $a_i = \theta_{i,x}$ and $\theta_i$ solves the first-best HJB
  equation with $\phi = 1$. Both projects are active on
  $[\ux_1, \bar x]$, with $a_1(\ux_1) = 0 < a_2(\ux_1)$. Let
  $x^{\ddagger} \ceq \inf\{x \in (\ux_1, \bar x] : a_1(x) \geq
  a_2(x)\}$, with $x^{\ddagger} \ceq \bar x$ if the set is empty. Then
  $a_1 < a_2$ on $(\ux_2, x^{\ddagger})$, as $a_1 = 0 < a_2$ on
  $(\ux_2, \ux_1]$. If $x^{\ddagger} < \bar x$, the argument in the
  proof of Proposition~\ref{Prop:compfb}\ref{it:pcfb02}, applied to
  $D \ceq \theta_2 - \theta_1$ and $T \ceq r_1\theta_1 - r_2\theta_2$,
  shows that $a_1 > a_2$ on $(x^{\ddagger}, \bar x]$.
\end{proof}

\subsection*{Proof of Proposition \ref{Prop:Principal}}

\begin{proof}[Proof of parts~\ref{it:p02} and \ref{it:p03}]
  Both parts follow by Proposition~\ref{Prop:compfb}\ref{it:pcfb04}.
\end{proof}

\begin{proof}[Proof of part~\ref{it:p01}]
  The proof of Proposition~\ref{Prop:compfb}\ref{it:pcfb04} gives
  $\pi < v$ on $(\ux, \bar x)$, and on $(\ux^{FB}, \ux]$ we have
  $\pi(x) = \uv < v(x)$. For the upper bound, using the
  first-best policy in 
  \eqref{prpr} gives
  \begin{equation*}
    \pi(x) \geq v(x) - \frac{\phihat - \phi}{2}\,
    \E_x\!\left[\int_0^{\tau^{FB}} e^{-rt} a^{FB}(X_t)^2\, dt\right],
  \end{equation*}
  where $\tau^{FB}$ is the first-best stopping time. By the definition
  of $v$,
  \begin{equation*}
    \frac{\phi}{2}\, \E_x\!\left[\int_0^{\tau^{FB}} e^{-rt}
      a^{FB}(X_t)^2\, dt\right] = \E_x\!\left[ e^{-r \tau^{FB}} \big(
      b\, \mathbf{1}_{\{X_{\tau^{FB}} = \bar x\}} + \uv\,
      \mathbf{1}_{\{X_{\tau^{FB}} < \bar x\}} \big) \right] - v(x)
    \leq b - v(x).
  \end{equation*}
  Combining,
  $v(x) - \pi(x) \leq \frac{\phihat - \phi}{\phi}\,(b - v(x)) =
  (\kappa - 1)(b - v(x)) = \phi \eta r \sigma^2 (b - v(x))$.
\end{proof}

\subsection*{Proof of Proposition \ref{Prop:Coupling}}
Set $Y_t \ceq X^{FB}_t - X_t$ on $[0, \tau \wedge \tau^{FB}]$. If
$\tau = 0$, the claims are immediate; so let $\tau > 0$. The Brownian
increments cancel, so $Y$ is differentiable with $Y_0 = 0$
and $\dot Y_t = a^{FB}(X^{FB}_t) - a(X_t)$. If $Y_t = 0$, then
$X^{FB}_t = X_t \in [\ux, \bar x]$ and
$\dot Y_t = a^{FB}(X_t) - a(X_t) > 0$ by
Proposition~\ref{Prop:Principal}\ref{it:p03}. Hence $Y$ becomes
positive immediately after $t = 0$ and cannot return to zero, which
would require a nonpositive derivative at the point of return; so
$Y_t > 0$ on $(0, \tau \wedge \tau^{FB}]$. Two consequences
follow. First, $\tau^{FB} \geq \tau$ unless the first best completes
before $\tau$, because termination before $\tau$ would require
$X^{FB}_t = \ux^{FB} < \ux \leq X_t < X^{FB}_t$. This gives
$X^{FB}_{t \wedge \tau^{FB}} > X_t$ for $0 < t < \tau$. Second, on
$\{X_\tau = \bar x\}$, if $\tau^{FB} \geq \tau$ then
$X^{FB}_\tau > X_\tau = \bar x$, which is impossible; so the first
best completed strictly before $\tau$. On
$\{X^{FB}_{\tau^{FB}} = \ux^{FB}\}$, the first best did not complete
before $\tau$, so $\tau^{FB} \geq \tau$, with equality impossible
since $X^{FB}_\tau > X_\tau \geq \ux > \ux^{FB}$; thus
$\tau < \tau^{FB}$, and $X_\tau = \bar x$ is excluded by the previous
case, so the second best terminated at $\tau$. The inclusion of the
completion events gives the weak ranking of completion
probabilities. It is strict because the second best terminates with
positive probability, and at that time the first best has either
completed the project or is still active, in which case it completes
with positive probability.

\subsection*{Proof of Proposition~\ref{Prop:Endo}}
Let $D \ceq v - \pi \geq 0$, where the inequality holds by a mimicking
argument. Both continuation regions contain $0$ as the project is
worth starting by assumption. First, $\bar x \leq \bar x^{FB}$:
otherwise $\bar x^{FB} \in (\ux, \bar x)$, so
$\pi(\bar x^{FB}) > b(\bar x^{FB}) = v(\bar x^{FB})$, contradicting
$D \geq 0$. Next, $\bar x < \bar x^{FB}$: if $\bar x = \bar x^{FB}$,
then value matching and smooth pasting give
$D(\bar x) = D_x(\bar x) = 0$, and the two HJB equations at $\bar x$
give
$\frac{1}{2}\sigma^2 (\pi_{xx} - v_{xx})(\bar x) = \frac{1}{2}
b_x(\bar x)^2 (\phi^{-1} - \phihat^{-1}) > 0$, where the inequality
uses $b_x > 0$, so $D_{xx}(\bar x) < 0$ and $D < 0$ to the left of
$\bar x$, a contradiction. Finally, $\ux > \ux^{FB}$ follows from the
proof of Proposition~\ref{Prop:compfb}\ref{it:pcfb04}, which does not
use the upper endpoint of the continuation region, since
$\phihat > \phi$.

\subsection*{Proof of Proposition \ref{Prop:Spot}}
Under $\Gamma$, the agent receives
$d\Gamma_s = \frac{\phi}{2} A_s^2\, ds + dCE_s$ until $\tau$ and
nothing thereafter. Integrating by parts, with $CE$ constant after
$\tau$, gives for every $t \geq 0$, almost surely,
\begin{equation*}
  \int_t^\tau e^{-r(s-t)} \Big( \frac{\phi}{2} A_s^2\, ds + dCE_s
  \Big) = \int_t^\infty e^{-r(s-t)} C_s\, ds - CE_t .
\end{equation*}
The two payment streams thus differ in present value by $CE_t$ along
every path, and $CE_t$ is known at $t$. Hence, along every path,
savings plus the present value of receipts are the same under $\Gamma$
with savings $S$ as under $C$ with savings $S - CE_t$, so by the
intertemporal budget constraint the two problems have the same
feasible policies and payoffs, which is the first claim. By
Lemma~\ref{lem:AWE}, the agent's continuation utility under $\Gamma$
with savings $S$ is therefore $e^{-\eta r (S - CE_t)} W_t$, and his
outside option is worth $S$. Since the certainty equivalent of the
former is $S$ by \eqref{eq:CE}, the participation constraint binds. At
$t = 0$, $CE_0 = 0$ as \eqref{ir} binds, so the agent's problem under
$\Gamma$ from zero savings is his problem under $C$. He thus chooses
the same effort and consumption, and as the payments have the same
present value along every path, both parties' expected payoffs are the
same under $\Gamma$ as under $C$. Finally, following the recommended
effort and consuming $C_s$ implies that the agent's savings satisfy
$dS_s = r S_s\, ds + d\Gamma_s - C_s\, ds = r(S_s - CE_s)\, ds +
dCE_s$ by \eqref{eq:spot} and \eqref{cons}, and after $\tau$ with
$d\Gamma_s = dCE_s = 0$ and $C_s = r\, CE_s$, so
$S_s - CE_s = (S_0 - CE_0) e^{rs} = 0$ for all $s$.

\subsection*{Proof of Lemma~\ref{lem:NecessarySuff}}
\emph{Resets if $L < \pi^{NR}(x^R; \phihat) - \uv$.} Modify the optimal no-reset
contract so that, when progress first reaches $\ux$, the project is
reset rather than terminated, and the optimal no-reset contract is
followed from $x^R$ thereafter. This changes the principal's payoff
by
\begin{equation*}
  \mathbb{E}\big[ e^{-r\tau} \big(\pi^{NR}(x^R; \phihat) - L - \uv\big)
  \mathbf{1}_{\{X_\tau = \ux\}} \big] > 0,
\end{equation*}
as $\ux$ is reached with positive probability. Thus, contracts without
resets are suboptimal.

\emph{No resets if $L \geq \pi^{NR}(x^R; \phihat) - \uv$.} We use a verification
argument. Write $\pi^{NR}$ for $\pi^{NR}(\cdot;\phihat)$, extended by
$\uv$ below
$\ux$ and by $b$ above $\bar x$. Then $\pi^{NR} \geq \uv$, and
$\pi^{NR}$ is continuously differentiable on $(-\infty, \bar x]$ and
satisfies
\begin{equation*}
  r \pi^{NR}(x) \geq \max_a \Big\{ -\frac{\phihat}{2} a^2
  + a \pi^{NR}_x(x) + \frac{1}{2}\sigma^2 \pi^{NR}_{xx}(x) \Big\},
\end{equation*}
with equality on $[\ux, \bar x]$. Consider the discounted value along
the path net of the costs paid so far,
\begin{equation*}
  M_t \ceq e^{-rt}\pi^{NR}(X_t) - \int_0^t e^{-rs}
  \frac{\phihat}{2} A_s^2\, ds.
\end{equation*}
By It\^o's formula, between resets,
\begin{equation*}
  dM_t = e^{-rt} \Big[ -r \pi^{NR}(X_t) + A_t \pi^{NR}_x(X_t)
  + \frac{1}{2}\sigma^2 \pi^{NR}_{xx}(X_t) - \frac{\phihat}{2} A_t^2
  \Big] dt + e^{-rt} \pi^{NR}_x(X_t) \sigma\, dZ_t,
\end{equation*}
where the drift is nonpositive by the inequality above, and zero
under the optimal no-reset contract, while the stochastic integral
is a martingale as $\pi^{NR}_x$ is bounded. Hence $M$ is a
supermartingale between resets for any effort process, and a
martingale under the optimal no-reset contract.

Now $\pi^{NR}(x^R) - L \leq \uv \leq \pi^{NR}(x)$ for all $x$, so at
every reset time $T_i < \tau$ the jump
$\pi^{NR}(x^R) - \pi^{NR}(X_{T_i-})$ in $\pi^{NR}(X_t)$ is at most
$L$. Thus, for any contract, $M_t - \sum_{T_i \leq t} e^{-rT_i} L$ is
a supermartingale. Since $\pi^{NR}(X_\tau)$ dominates the stopping
payoff, the principal's payoff is at most $M_0 = \pi^{NR}(x)$, which
is attained without resets. Hence the optimal contract does not use
resets.

\emph{Comparison to first best.} The two cases apply to the first
best with $\phi$ in place of $\phihat$, and
$\pi^{NR}(x^R; \phi) \geq \pi^{NR}(x^R; \phihat)$ by the mimicking
argument in the proof of Proposition~\ref{Prop:compfb}\ref{it:pcfb04}.

\subsection*{Proof of Proposition~\ref{Prop:reset}}

\begin{proof}[Proof of part~\ref{it:r03}]
  The reset problem is the no-reset problem with salvage value
  ${\pi(x^R) - L}$, and likewise for the first best with $v(x^R) - L$.
  As $x^R \to \bar x$, both salvage values tend to $b - L$, and
  $\ux > \ux^{FB}$ for the no-reset problems with salvage value
  $b - L$ by Proposition~\ref{Prop:Principal}\ref{it:p02}. Since the
  termination threshold depends continuously on the salvage value by
  autonomy of the ODE (cf.\ proof of
  Proposition~\ref{Prop:compfb}\ref{it:pcfb01}) and continuous
  dependence on initial data (\citealp{teschl2012ordinary},
  Theorem~2.8), the strict inequality persists for $x^R$ close to
  $\bar x$.
\end{proof}

\begin{proof}[Proof of part~\ref{it:r04}]
  We first rewrite the ODE~\eqref{eq:ODE} in terms of the level $y = \pi$.
  \begin{lemma}\label{lem:level}
    Let $\pi$ solve $r\pi = \frac{1}{2\phi}\pi_x^2 + \frac{1}{2}\sigma^2
    \pi_{xx}$ on $[\ux, \bar x]$ with $\pi(\ux) = s > 0$ and
    $\pi_x(\ux) = 0$. Then $\pi_x(x)^2 = F(\pi(x); s, \phi)$ for all
    $x \in [\ux, \bar x]$, where
    \begin{equation*}
      F(y; s, \phi) \ceq 2 r \phi \Big[ y - s + \Big( s -
      \frac{\phi \sigma^2}{2} \Big) \Big( 1 - e^{-2(y - s)/(\phi
        \sigma^2)} \Big) \Big],
    \end{equation*}
    and for $\ux \leq x_1 \leq x_2 \leq \bar x$,
    \begin{equation*}
      x_2 - x_1 = \int_{\pi(x_1)}^{\pi(x_2)} \frac{dy}{\sqrt{F(y; s,
          \phi)}}.
    \end{equation*}
    Moreover, $F(y; s, \phi) \leq \min\{2 r \phi y, 2 r y^2 /
    \sigma^2\}$, and $F(s + \theta; s, \phi)$ increases in
    $\theta \geq 0$, $s$, and $\phi$.
  \end{lemma}
  \begin{proof}
    By the argument of Proposition~\ref{Prop:FirstBest}\ref{it:d02},
    $\pi$ is increasing on $(\ux, \bar x]$, so $F \ceq \pi_x^2$ can be
    viewed as a function of the level $y = \pi$. Then
    $dF/dy = 2 \pi_{xx} = \frac{4r}{\sigma^2} y - \frac{2}{\phi
      \sigma^2} F$, a linear equation with $F(s) = 0$, whose solution
    is the display. The distance formula is
    $dx/dy = 1/\pi_x = F^{-1/2}$. From the equation, $dF/dy \geq 0$
    once $F \leq 2 r \phi y$, which holds as $\pi_{xx} > 0$;
    differentiating the display in $s$ at fixed $\theta = y - s$ gives
    $2 r \phi (1 - e^{-2 \theta / (\phi \sigma^2)}) > 0$; for
    $\phi_1 < \phi_2$, the two solutions agree to first order at
    $y = s$ with $F_1'' < F_2''$ there, and at any later meeting point
    $F_1' < F_2'$ by the equation, so $F_1 < F_2$ on $(s,
    \infty)$. Finally, $F$ is decreasing in $s$ at fixed $y$ (the
    derivative is
    $-\frac{4r}{\sigma^2} s\, e^{-2(y - s)/(\phi \sigma^2)}$), and
    $F(y; 0, \phi) = 2 r \phi [y - \frac{\phi \sigma^2}{2} (1 -
    e^{-2y/(\phi \sigma^2)})] \leq 2 r y^2 / \sigma^2$ by
    $1 - e^{-z} \geq z - z^2/2$. \hfill $\square$
  \end{proof}

  Since $\pi^{NR}(\cdot; \phihat)$ increases continuously from $\uv$ to
  $b$ and $L < b - \uv$, restart points with
  $0 < \pi^{NR}(x^R; \phihat) - \uv - L < \varepsilon$ exist for every
  $\varepsilon > 0$. Fix one and let
  $m \ceq \pi^{NR}(x^R; \phihat) - \uv - L \in (0, \varepsilon)$. Write
  $F(y; s) \ceq F(y; s, \phihat)$ and $F^{FB}(y; s) \ceq F(y; s,
  \phi)$, and for $s \in [\uv, b - L]$ let
  \begin{equation*}
    G(s) \ceq \int_{s + L}^{b} \frac{dy}{\sqrt{F(y; s)}} = \int_{L}^{b
      - s} \frac{d\theta}{\sqrt{F(s + \theta; s)}}.
  \end{equation*}
  By Lemma~\ref{lem:level} the integrand is decreasing in $s$, and
  $F(s + \theta; s) \leq 2 r \phihat b$ for $s + \theta \leq b$, so
  $G(s) - G(s') \geq (s' - s)/\sqrt{2 r \phihat b}$ for $s < s'$.

  The reset problem is the no-reset problem with salvage value $\tilde
  v \ceq \pi(x^R) - L \in [\uv, b - L]$, so by Lemma~\ref{lem:level},
  $\bar x - x^R = G(\tilde v)$ and $x^R - \ux^R = \int_{\tilde
    v}^{\tilde v + L} dy / \sqrt{F(y; \tilde v)}$. Let
  $x^\dagger \in (\ux, x^R)$ be the point where
  $\pi^{NR}(x^\dagger; \phihat) = \uv + L$, which exists as
  $m > 0$. Then $\bar x - x^\dagger = G(\uv)$, so
  $x^R - x^\dagger = G(\uv) - G(\tilde v) \geq (\tilde v - \uv) /
  \sqrt{2 r \phihat b}$, while by convexity and Lemma~\ref{lem:level},
  $m = \pi^{NR}(x^R; \phihat) - \pi^{NR}(x^\dagger;
  \phihat) \geq \sqrt{F(\uv + L; \uv)}\, (x^R - x^\dagger) \geq
  \sqrt{F(L; 0)}\, (x^R - x^\dagger)$. Hence
  $\tilde v - \uv \leq C m$ with
  $C \ceq \sqrt{2 r \phihat b / F(L; 0)}$, and by the bound
  $F(y; s) \leq 2 r y^2 / \sigma^2$,
  \begin{equation*}
    \ux^R \leq x^R - \frac{\sigma}{\sqrt{2r}} \int_{\tilde v}^{\tilde
      v + L} \frac{dy}{y} \leq \bar x - \frac{\sigma}{\sqrt{2r}} \ln
    \Big( 1 + \frac{L}{\uv + C m} \Big),
  \end{equation*}
  which tends to $-\infty$ as $\uv \to 0$ and $m \to 0$.

  For the first best, the reset problem is the no-reset problem with
  salvage value $\tilde v^{FB} \ceq v(x^R) - L$, so $\ux^{R,FB} = \bar
  x - \int_0^{b - \tilde v^{FB}} d\theta / \sqrt{F^{FB}(\tilde v^{FB}
    + \theta; \tilde v^{FB})}$, and the integral is decreasing in
  $\tilde v^{FB}$ by Lemma~\ref{lem:level}. It thus suffices to bound
  $\tilde v^{FB}$ away from zero for $\uv$ small. Since the no-reset
  contract is feasible and $v^{NR}$ is increasing, $\tilde v^{FB} \geq
  \ell - L$, where $\ell \ceq v^{NR}(x^\dagger; \phi)$. Applying
  Lemma~\ref{lem:level} to $v^{NR}$ and to $\pi^{NR}$ on $[x^\dagger,
  \bar x]$ gives $\int_\ell^b dy / \sqrt{F^{FB}(y; \uv)} = \bar x -
  x^\dagger = \int_{\uv + L}^b dy / \sqrt{F(y; \uv)}$, so, as $F^{FB}
  \leq F$,
  \begin{equation*}
    \int_{\uv + L}^{\ell} \frac{dy}{\sqrt{F^{FB}(y; \uv)}} = \int_{\uv
      + L}^{b} \Big[ \frac{1}{\sqrt{F^{FB}(y; \uv)}} -
    \frac{1}{\sqrt{F(y; \uv)}} \Big] dy \eqc H(\uv).
  \end{equation*}
  The function $H$ is continuous on $[0, b - L)$ with $H(0) > 0$, so
  $H \geq h > 0$ on $[0, \varepsilon_0]$ for some $\varepsilon_0, h >
  0$. As $F^{FB}(y; \uv)$ is increasing in $y$ and $F^{FB}(\uv + L;
  \uv) \geq F^{FB}(L; 0)$, the left side is at most $(\ell - \uv - L)
  / \sqrt{F^{FB}(L; 0)}$. Hence $\tilde v^{FB} \geq \ell - L \geq h
  \sqrt{F^{FB}(L; 0)} \eqc c > 0$ for all $\uv \leq \varepsilon_0$,
  and $\ux^{R,FB} \geq \bar x - \int_0^{b - c} d\theta /
  \sqrt{F^{FB}(c + \theta; c)}$, a bound independent of $\uv$ and
  $x^R$. Choosing $\varepsilon \leq \varepsilon_0$ small enough that
  the bound on $\ux^R$ falls below it for $\uv < \varepsilon$ and
  $m < \varepsilon$ gives $\ux^R < \ux^{R,FB}$.
\end{proof}

\begin{proof}[Proof of part~\ref{it:r06}]
  Both policies are active on $(\max\{\ux^R, \ux^{R,FB}\}, \bar x]$,
  where $\pi$ and $v$ solve their ODEs, $v \geq \pi$ since the first
  best can mimic the second best, and $\pi_{xx} > 0$ by
  Proposition~\ref{Prop:SecondBest}\ref{it:s02} applied to the
  no-reset problem with salvage value $\tilde v$. At any point where
  $a = a^{FB} > 0$, we have $\pi_x = \phihat a$ and $v_x = \phi a$, and
  subtracting the ODEs gives
  \begin{equation*}
    \frac{\sigma^2}{2} (v_{xx} - \pi_{xx}) = r (v - \pi) +
    \frac{\phihat - \phi}{2} a^2 > 0,
  \end{equation*}
  so $a^{FB}_x = \phi^{-1} v_{xx} > \phi^{-1} \pi_{xx} > \phihat^{-1}
  \pi_{xx} = a_x$. Thus $a^{FB}$ crosses $a$ only from below, and
  hence at most once. If $\ux^R > \ux^{R,FB}$, then
  $a^{FB}(\ux^R) > 0 = a(\ux^R)$, so a crossing would have $a^{FB}$
  crossing from above, and $a < a^{FB}$ on $(\ux^R, \bar x]$. If
  $\ux^R < \ux^{R,FB}$, then $a(\ux^{R,FB}) > 0 = a^{FB}(\ux^{R,FB})$,
  so $a > a^{FB}$ immediately to the right of $\ux^{R,FB}$, and the
  claim holds with $x^*$ the crossing point, or $x^* \ceq \bar x$ if
  there is none.
\end{proof}

\subsection*{Proof of Lemma~\ref{Lem:IcRisk}}
Under $(A, Q)$, the process
$M_t \ceq \int_0^t e^{-rs} u(A_s, C_s)\, ds + e^{-rt} W_t$ is the
conditional expectation of the agent's lifetime utility given public
history up to $t$, and hence it is a martingale with respect to the
filtration generated by $Z$ and $N$. By the martingale representation
theorem (see, e.g., Theorem 11.2.8.1 of
\citealp{jeanblanc2009mathematical}), there exist progressively
measurable processes $\beta$ and $\psi$ such that
\begin{equation*}
  dM_t = e^{-rt} \beta_t (-\eta r W_t)\, \sigma\, dZ_t
  + e^{-rt} \psi_t (-\eta r W_{t-}) \big( dN_t - (\lambda_0 + \lambda_r
  Q_t)\, dt \big).
\end{equation*}
By the definition of $M$, also
$dM_t = e^{-rt} \big( u(A_t, C_t) - r W_t \big)\, dt + e^{-rt}\,
dW_t$. Equating the two expressions and substituting
$\sigma\, dZ_t = dX_t - (A_t + g Q_t)\, dt$ gives \eqref{eq:W-risk}.

The Euler equation $r W_t = u(A_t, C_t)$ and \eqref{cons} follow by
Lemma~\ref{lem:AWE} as in the discussion after
Lemma~\ref{Lem:IcNoS}. Consider deviating to $(\hat a, \hat q)$ at
time $t$ with consumption unchanged. By \eqref{eq:W-risk}, this
changes the drift of $W$ by
$(-\eta r W_t) \big[ \beta_t (\hat a - A_t) + (g \beta_t + \lambda_r
\psi_t)(\hat q - Q_t) \big]$. Thus, incentive compatibility requires that
$(A_t, Q_t)$ solve (a.s.\ and for a.e.\ $t$)
\begin{equation*}
  \max_{a \geq 0,\, \hat q \in \{0,1\} } u(\hat a, C_t) + (-\eta r W_t) \big[ \beta_t \hat a
  + (g \beta_t + \lambda_r \psi_t)\, \hat q \big].
\end{equation*}
The objective is linear in $\hat q$, which gives
\eqref{eq:IC-risk}. It is strictly concave in $\hat a$ with derivative
$\eta \phi \hat a\, u(\hat a, C_t) + \beta_t (-\eta r W_t)$, which is
negative for all $\hat a \geq 0$ if $\beta_t < 0$, and vanishes at
$\hat a = A_t$ if and only if $\phi A_t = \beta_t$ if
$\beta_t \geq 0$, using $u(A_t, C_t) = r W_t$. This gives \eqref{IC}.

\subsection*{Proof of Proposition \ref{Prop:MartingaleLevy}}
Under obedience, i.e., when the agent obeys $(A, Q)$ and consumes
  $C$, \eqref{eq:W-risk} and \eqref{cons} make $W$ a stochastic
  exponential with bounded coefficients, hence a martingale, and
  integrating \eqref{eq:W-risk} against $e^{-r(s-t)}$ shows that
  ${W_t = \E[\int_t^\infty e^{-r(s-t)} u(A_s, C_s)\, ds \mid \mathcal
  F_t]}$. So $W$ is the agent's continuation utility under
  obedience. Moreover, $W_0 = u(A_0, C_0)/r$ is finite by
  \eqref{cons}. For $t>\tau$, $A_t = 0$ and $C_t$ is constant, so
  $W_t$ is constant by \eqref{cons}, hence $\beta_t = \psi_t = 0$ by
  \eqref{eq:W-risk}, and $C_t = r\, CE(W_\tau)$.

  Fix a deviation $(\hat A, \hat C, \hat Q)$ whose savings $S$ satisfy
  the no-Ponzi condition and whose expected payoff is finite; as
  $u < 0$ and $W_0$ is finite, no other deviation can be profitable.
  Then progress follows
  $dX_t = (\hat A_t + g \hat Q_t)\, dt + \sigma\, dZ_t$, $W$ still
  satisfies \eqref{eq:W-risk} as all its terms are functionals of the
  public history, and breakdowns arrive at rate
  $\lambda_0 + \lambda_r \hat Q_t$; all expectations below are under
  this deviation. Note that $\tau < \infty$ almost surely: the drift
  of $X$ is nonnegative, so $X_t \geq X_0 + \sigma Z_t$ for
  $t \leq \tau$, and as a Brownian motion exceeds every level in
  finite time, $X$ reaches $\bar x$ in finite time unless termination
  or a breakdown occurs first.

For $t \geq 0$, let $J_t \ceq e^{-\eta r S_t} W_t$ and
\begin{equation}\label{eq:Ydef}
  Y_t \ceq \int_0^t e^{-rs} u(\hat A_s, \hat C_s)\, ds
  + e^{-rt} J_t.
\end{equation}
This is the expected payoff, conditional on $\mathcal F_t$, of
following the deviation on $[0, t]$ and obeying the contract from $t$
on: as the contract's terms depend only on the public history, $W_t$
is the utility of obeying from $t$ on with zero savings, and with
savings $S_t$ it is $J_t$ by Lemma~\ref{lem:AWE}. In
particular, $Y_0 = W_0$ is the payoff from obedience, and the payoff
from the deviation is $\E[Y_\infty]$, where
$Y_\infty \ceq \int_0^\infty e^{-rs} u(\hat A_s, \hat C_s)\, ds$. We
will show that $\E[Y_\infty] \leq W_0$.

\emph{Step 1: Ruling out one-shot deviations.} By It\^o's
formula, with $dS_t = (r S_t + C_t - \hat C_t)\, dt$ and $dW_t$ from
\eqref{eq:W-risk}, the drift of $Y$ is
\begin{equation*}
  e^{-rt} e^{-\eta r S_t} \big[ F_t(\hat A_t, \hat C_t - r S_t, \hat Q_t)
  - F_t(A_t, C_t, Q_t) \big],
\end{equation*}
where
\begin{equation*}
  F_t(a, c, q) \ceq u(a, c)
  + (-\eta r W_t) \big[ \beta_t a + (g \beta_t + \lambda_r \psi_t)
  q - c \big].
\end{equation*}
The function $F_t$ is concave in $(a, c)$ and linear in $q$, and
the conditions of Lemma~\ref{Lem:IcRisk} are its first-order
conditions at $(A_t, C_t, Q_t)$:
$\partial F_t/\partial c = -\eta u(A_t, C_t) + \eta r W_t = 0$ by
\eqref{cons};
$\partial F_t/\partial a = (-\eta r W_t)(\beta_t - \phi A_t)$, which
is zero if $\beta_t \geq 0$ and negative at $A_t = 0$ if
$\beta_t < 0$, by \eqref{IC}; and
$\partial F_t/\partial q = (-\eta r W_t)(g \beta_t + \lambda_r \psi_t)$,
so $Q_t$ maximizes by \eqref{eq:IC-risk}. Hence $Y$ is a local
supermartingale, and a local martingale under obedience.

\emph{Step 2: Bounding continuation payoffs.} Define the constants
$\bar\beta \ceq \|\beta\|_\infty$, $\bar\psi \ceq \|\psi\|_\infty$,
$H \ceq \bar\psi (\lambda_0 + \lambda_r) + \frac{1}{2} \eta r \sigma^2
\bar\beta^2$, and $K \ceq e^{\eta H}$. We show that, for every
stopping time $\rho$,
\begin{equation}\label{eq:UIbound}
  \E\Big[ \int_\rho^\infty e^{-r(s-\rho)} u(\hat A_s, \hat C_s)\, ds
  \,\Big|\, \mathcal F_\rho \Big]
  \leq \frac{J_\rho}{K}.
\end{equation}
Let $V_t \ceq CE(J_t) = S_t + CE(W_t)$ be the wealth
equivalent of obeying from $t$ on, and let
$\xi_t \ceq \hat C_t - \frac{\phi}{2} \hat A_t^2$. By It\^o's formula,
using \eqref{cons}, \eqref{eq:W-risk}, and
$dS_t = (r S_t + C_t - \hat C_t)\, dt$,
\begin{align}
  & dV_t = (r V_t - \xi_t + D_t)\, dt
  + \sigma \beta_t\, dZ_t
  - \tfrac{1}{\eta r} \ln(1 - \eta r \psi_t)\, dN_t,
  \quad\text{where} \label{eq:wealth} \\
  & D_t \ceq [ \beta_t \hat A_t - \tfrac{\phi}{2} \hat A_t^2 ]
  - [ \beta_t A_t - \tfrac{\phi}{2} A_t^2 ]
  + (g \beta_t + \lambda_r \psi_t)(\hat Q_t - Q_t)
  - \psi_t (\lambda_0 + \lambda_r \hat Q_t)
  + \tfrac{1}{2} \eta r \sigma^2 \beta_t^2. \notag
\end{align}
The first three terms are nonpositive by \eqref{IC} and
\eqref{eq:IC-risk}, the last two at most $H$, so $D_t \leq H$.

Fix a stopping time $\rho$ and $T > \rho$. Integrating
\eqref{eq:wealth} against $e^{-r(s-\rho)}$ over $[\rho, T]$, and
using ${D_s \leq H}$ and ${\ln(1 - \eta r \psi_s) \geq 0}$, gives
\begin{equation*}
  \int_\rho^T e^{-r(s-\rho)} \xi_s\, ds
  \leq V_\rho - e^{-r(T-\rho)} V_T + \frac{H}{r}
  + \sigma \int_\rho^T e^{-r(s-\rho)} \beta_s\, dZ_s.
\end{equation*}
Let ${p_T \ceq 1 - e^{-r(T-\rho)}}$. By Jensen's inequality for the
density ${r e^{-r(s-\rho)}/p_T}$ on $[\rho, T]$,
\begin{equation*}
  \int_\rho^T e^{-r(s-\rho)} u(\hat A_s, \hat C_s)\, ds
  = -\frac{p_T}{\eta r} \int_\rho^T \frac{r e^{-r(s-\rho)}}{p_T}
  e^{-\eta \xi_s}\, ds
  \leq -\frac{p_T}{\eta r} \exp\big\{ -\tfrac{\eta r}{p_T} \int_\rho^T
  e^{-r(s-\rho)} \xi_s\, ds \big\},
\end{equation*}
and the previous display bounds the integral in the exponent. Now let
$T \to \infty$. Then ${p_T \to 1}$, and ${e^{-r(T-\rho)} V_T \to 0}$,
since $CE(W_T)$ is constant for $T > \tau$ and ${e^{-rT} S_T \to 0}$ by
the no-Ponzi condition. The stochastic integral converges, as $\beta$
is bounded. Hence
\begin{equation*}
  \int_\rho^\infty e^{-r(s-\rho)} u(\hat A_s, \hat C_s)\, ds
  \leq \frac{J_\rho}{K}\, e^{-\eta r \mathcal B_\rho},
  \quad\text{where}\quad
  \mathcal B_\rho \ceq \sigma \int_\rho^\infty e^{-r(s-\rho)}
  \beta_s\, dZ_s.
\end{equation*}
Take conditional expectations on both sides. As $\beta$ is bounded,
${\E[\mathcal B_\rho \mid \mathcal F_\rho] = 0}$, so the conditional
Jensen inequality gives
${\E[e^{-\eta r \mathcal B_\rho} \mid \mathcal F_\rho] \geq 1}$, and
\eqref{eq:UIbound} follows as $J_\rho < 0$.

\emph{Step 3: Removing the localization.} Let
$\rho_n \ceq \inf\{t : |J_t| \geq n\} \wedge n$, so that
$\rho_n \to \infty$ as $J$ is finite and right-continuous. On
$[0, \rho_n]$, the integrands of the stochastic
integrals in \eqref{eq:Ydef} are bounded, so $Y$ stopped at $\rho_n$ is a
supermartingale by Step~1, and $\E[Y_{\rho_n}] \leq Y_0 = W_0$. As
$n \to \infty$, $\int_0^{\rho_n} e^{-rs} u(\hat A_s, \hat C_s)\, ds
\to Y_\infty$ in $L^1$ by dominated convergence, since $Y_\infty$ is
integrable and $u < 0$. The second term vanishes in $L^1$: taking
absolute values in \eqref{eq:UIbound},
\begin{equation*}
  \E\big[ e^{-r\rho_n} |J_{\rho_n}| \big]
  \leq K\, \E\Big[ \int_{\rho_n}^\infty e^{-rs}
  |u(\hat A_s, \hat C_s)|\, ds \Big] \to 0,
\end{equation*}
again by dominated convergence. Hence
$\E[Y_\infty] = \lim_n \E[Y_{\rho_n}] \leq W_0$.

\subsection*{Proof of Proposition~\ref{Prop:OptimalContract}}
Throughout, let $\Lambda \ceq r + \lambda_0 + \lambda_r$,
$\mu \ceq \phi g/\lambda_r$, and $z \ceq \pi_x(x)$. Maximizing in
\eqref{hjb:risk} first over $(a, \psi)$ given the risk mode and then
over the mode, the HJB equation reads
\begin{equation}\label{eq:hjb-sb-gap}
  r\pi = \frac{\sigma^2}{2} \pi_{xx} + \max\{H^R, H^S\},
\end{equation}
where, defining the short-hands
$h(\psi) \ceq \psi + \frac{1}{\eta r}\ln(1 - \eta r\psi)$ and
$\Psi(a) \ceq -\lambda_0 h(-\mu a)$, 
\begin{align*}
  H^R &\ceq \max_{a \geq 0,\, \psi \geq -\mu a} \Big\{ -\frac{\phihat}{2} a^2
        + (a + g) z - (\lambda_0 + \lambda_r)\big(\pi - \uv - h(\psi)\big) \Big\} \\
  & \; = \frac{z^2}{2\phihat} + g z - (\lambda_0 + \lambda_r)(\pi - \uv), \\
  H^S &\ceq \max_{a \geq 0,\, \psi \leq -\mu a} \Big\{ -\frac{\phihat}{2} a^2
        + a z - \lambda_0 \big(\pi - \uv - h(\psi)\big) \Big\}
        = G(z) - \lambda_0 (\pi - \uv),
\end{align*}
and
$G(z) \ceq \max_{a \geq 0} \{ -\frac{\phihat}{2} a^2 + a z - \Psi(a)
\}$. (The closed form of $H^R$ uses $z \geq 0$, which
part~\ref{it:brp01} establishes using it only at $z = 0$.) To verify
the expressions for $H^R$ and $H^S$, note that $h$ is concave with
$h(0) = 0$ and $h'(\psi) = 1 - (1 - \eta r\psi)^{-1}$, so it is
maximized at $\psi = 0$. In the risky mode, $\psi = 0$ is feasible, so
$\psi = 0$ and $a = a^R(z) \ceq z/\phihat$ by the first-order
condition (if $z \geq 0$). In the safe mode, the constraint binds, as
$h$ is increasing on $\psi \leq 0$. Substituting $\psi = -\mu a$ gives
the maximization problem defining $G$, explaining $H^S$.

By \eqref{eq:hjb-sb-gap}, the risky mode is optimal if and only if
$D^{SB} \ceq H^R - H^S \geq 0$. Writing
$\Gamma(z) \ceq \frac{z^2}{2\phihat} + g z - G(z)$ for the gain in
the value of effort from the risky mode, we have
\begin{equation}\label{eq:DSBx}
  D^{SB} = \Gamma(z) - \lambda_r (\pi - \uv)
  \quad\text{and}\quad
  D^{SB}_x = \Gamma'(z)\, \pi_{xx} - \lambda_r z.
\end{equation}
We collect the properties of $\Gamma$ used below. As
$\Psi'(a) = \frac{\lambda_0 \eta r \mu^2 a}{1 + \eta r \mu a} \geq 0$
and
$\Psi''(a) = \frac{\lambda_0 \eta r \mu^2}{(1 + \eta r \mu a)^2} > 0$,
the objective in $G$ is strictly concave, and the effort $a^S(z)$ is
the unique solution to the first-order condition
${\phihat\, a - z + \Psi'(a) = 0}$, with $a^S(0) = 0$,
$0 < a^S(z) < z/\phihat$ for $z > 0$, and
$da^S/dz = 1/(\phihat + \Psi''(a^S)) \in (0, 1/\phihat]$, which is
increasing in $z$. Thus $\Gamma(0) = 0$, and by the envelope theorem
$\Gamma'(z) = z/\phihat + g - a^S(z) \geq g$ and
$\Gamma''(z) = 1/\phihat - da^S/dz \geq 0$. Moreover,
$z/\phihat - a^S = \Psi'(a^S)/\phihat$ and $\Psi' \leq \lambda_0 \mu$
give $\Gamma' \leq g(1 + \frac{\lambda_0}{\kappa \lambda_r}) \eqc
\bar\Gamma'$, and $\Psi''(a^S) \leq \Psi''(0)$ gives
$\Gamma'' \leq \frac{1}{\phihat} - \frac{1}{\phihat + \lambda_0 \eta r
  \mu^2} = \frac{\lambda_0 \eta r g^2}{\kappa(\kappa \lambda_r^2 +
  \lambda_0 \eta r \phi g^2)} \eqc \bar\Gamma''$.

\begin{proof}[Proof of part~\ref{it:brp01}]
  If $\pi_x(x) = 0$ at some $x$, then $z = 0$ and
  $D^{SB}(x) = -\lambda_r(\pi(x) - \uv) \leq 0$, so the safe mode is
  optimal there and, as $a^S(0) = 0$, \eqref{eq:hjb-sb-gap} gives
  $\frac{\sigma^2}{2} \pi_{xx}(x) = r\pi(x) + \lambda_0(\pi(x) - \uv)
  \geq r\uv > 0$. This is the analog of Lemma~\ref{lem-convex}, and
  $\pi_x > 0$ on $(\ux, \bar x]$ follows as in the proof of
  Proposition~\ref{Prop:FirstBest}\ref{it:d02}. For convexity,
  \eqref{eq:hjb-sb-gap} at $\ux$ gives
  $\pi_{xx}(\ux) = 2 r \uv/\sigma^2 > 0$. Suppose $\pi_{xx}$ has a zero
  in $(\ux, \bar x]$ and let $x^*$ be the first one. If
  $D^{SB}(x^*) = 0$, then $D^{SB}_x(x^*) = -\lambda_r \pi_x(x^*) < 0$
  by \eqref{eq:DSBx}, so the risky mode is strictly optimal
  immediately to the left of $x^*$; otherwise the mode that is
  strictly optimal at $x^*$ is so on a neighborhood. Differentiating
  \eqref{eq:hjb-sb-gap} from the left, where the mode $q$ is fixed,
  and using $\pi_{xx}(x^*) = 0$ gives
  $\frac{\sigma^2}{2} \pi_{xxx}(x^*-) = (r + \lambda_0 + \lambda_r q)\,
  \pi_x(x^*) > 0$, so $\pi_{xx} < 0$ immediately to the left of $x^*$,
  contradicting the choice of $x^*$. Hence $\pi_{xx} > 0$ on
  $[\ux, \bar x]$.
\end{proof}

\begin{proof}[Proof of part~\ref{it:brp02}]
  Value matching and smooth pasting give $D^{SB}(\ux) = 0$, and part
  \ref{it:brp01} and \eqref{eq:DSBx} give
  $D^{SB}_x(\ux) = g\, \pi_{xx}(\ux) > 0$, so the risky mode is
  strictly optimal just above $\ux$. At any zero $x$ of $D^{SB}$,
  \eqref{eq:DSBx} gives $\lambda_r(\pi(x) - \uv) = \Gamma(z)$.
  Substituting this and \eqref{eq:hjb-sb-gap} with $D^{SB} = 0$ into
  $D^{SB}_x$ in \eqref{eq:DSBx} gives
  \begin{equation}\label{eq:sb-gap}
    D^{SB}_x(x) = \Phi(z), \qquad
    \Phi(z) \ceq \frac{2}{\sigma^2}\, \Gamma'(z) \Big[ r \uv
    + \frac{\Lambda}{\lambda_r} \Gamma(z) - \frac{z^2}{2\phihat}
    - g z \Big] - \lambda_r z.
  \end{equation}
  We show below that $\Phi$ is decreasing under the hypothesis of the
  proposition. As $\Phi(0) = 2 g r \uv/\sigma^2 > 0$, $\Phi$ is then
  positive below its unique root and negative above it. The set
  $\{x \in (\ux, \bar x) : D^{SB}(x) < 0\}$ is open, and no component
  $(x_1, x_2)$ of it can have $x_2 < \bar x$:
  $D^{SB}(x_1) = D^{SB}(x_2) = 0$ would require
  $D^{SB}_x(x_1) \leq 0 \leq D^{SB}_x(x_2)$, i.e.,
  $\Phi(\pi_x(x_1)) \leq 0 \leq \Phi(\pi_x(x_2))$ with
  $\pi_x(x_1) < \pi_x(x_2)$ by part~\ref{it:brp01}. Hence the set is
  empty or an interval $(x_c^{SB}, \bar x)$ with $x_c^{SB} > \ux$;
  with $x_c^{SB} \ceq \bar x$ in the former case, the risky mode is
  optimal on $[\ux, x_c^{SB}]$ and the safe mode on
  $(x_c^{SB}, \bar x]$.

  It remains to show that $\Phi' < 0$ for $z \geq 0$. Let $B(z)$
  denote the bracket in \eqref{eq:sb-gap}, so that
  $\Phi' = \frac{2}{\sigma^2} [\Gamma'' B + \Gamma' B'] -
  \lambda_r$. First, $\Gamma(0) = 0$ and $\Gamma' \leq \bar\Gamma'$
  give $\Gamma(z) \leq \bar\Gamma' z$, so
  $B(z) \leq r\uv + (\Lambda \bar\Gamma'/\lambda_r - g) z -
  z^2/(2\phihat) \leq M_0 \ceq r\uv + \frac{\phihat}{2} (\Lambda
  \bar\Gamma'/\lambda_r - g)^2$. Second, we have
  $B' = \frac{\Lambda}{\lambda_r} \Gamma' - \frac{z}{\phihat} - g =
  \frac{r + \lambda_0}{\lambda_r} \Gamma' - a^S$ and
  $B'' = \frac{r + \lambda_0}{\lambda_r} \Gamma'' -
  \frac{da^S}{dz}$. As $\Gamma'' = 1/\phihat - da^S/dz$ and $da^S/dz$
  is increasing, $B''$ is nonincreasing, and $B''(0) \leq 0$ if and
  only if
  \begin{equation}\label{eq:sigma_cond}
    (r + \lambda_0)\, \lambda_0 \eta r \mu^2 \leq \lambda_r \phihat,
    \quad \text{i.e.,} \quad
    (r + \lambda_0)\, \lambda_0 \eta r \phi g^2 \leq \lambda_r^3 \kappa.
  \end{equation}
  Under \eqref{eq:sigma_cond}, $B'$ is nonincreasing, so
  $B' \leq B'(0) = (r + \lambda_0) g/\lambda_r$. Combining, and using
  $\Gamma'' \geq 0$ and $\Gamma' > 0$,
  $\Phi' \leq \frac{2}{\sigma^2} \big[ \bar\Gamma'' M_0 + \bar\Gamma'
  (r + \lambda_0) g/\lambda_r \big] - \lambda_r$, which is negative if
  and only if
  \begin{equation}\label{eq:suff2}
    \lambda_r^2 \sigma^2 > \frac{2 \lambda_0 \eta r g^2 \lambda_r M_0}
    {\kappa(\kappa \lambda_r^2 + \lambda_0 \eta r \phi g^2)}
    + 2 g^2 (r + \lambda_0) \Big( 1 + \frac{\lambda_0}{\kappa \lambda_r}
    \Big).
  \end{equation}
  Both \eqref{eq:sigma_cond} and \eqref{eq:suff2} hold for
  $\lambda_r$ or $\sigma$ large enough, as $\kappa = 1 + \phi \eta r
  \sigma^2$ and the right-hand side of \eqref{eq:suff2} stays bounded
  while the left-hand side grows without bound.
\end{proof}

\begin{proof}[Proof of part~\ref{it:brp05}]
  By \eqref{eq:hjb-sb-gap} and $\pi_{xx} > 0$, $r\pi \geq H^R$, so
  $z^2/(2\phihat) \leq r\pi + (\lambda_0 + \lambda_r)(\pi - \uv)$,
  which equals $M \ceq r\uv + \Lambda(b - \uv)$ at $\bar x$. Since
  $\Gamma(z) \leq \bar\Gamma' z$, \eqref{eq:DSBx} then implies that
  $D^{SB}(\bar x) \leq \bar\Gamma' \sqrt{2 \phihat M} - \lambda_r (b -
  \uv)$, which is negative for all $b$ large enough, as $M$ is linear
  in $b$. The safe mode is then strictly optimal near $\bar x$, so
  $x_c^{SB} < \bar x$ by part~\ref{it:brp02}.
\end{proof}

\begin{proof}[Proof of parts~\ref{it:brp04} and \ref{it:brp03}]
  As shown above, $\psi = 0$ and $a = \pi_x/\phihat$ in the risky
  mode, and $\psi = -\mu a$ and $a = a^S(\pi_x)$ in the safe
  mode. Selecting the risky mode at ties below $x_c^{SB}$, effort is
  $\pi_x/\phihat$ on $[\ux, x_c^{SB}]$ and $a^S(\pi_x)$ on
  $(x_c^{SB}, \bar x]$. Both are continuous and increasing, as
  $\pi_x$ is increasing and $da^S/dz > 0$. If
  $x_c^{SB} < \bar x$, effort drops at $x_c^{SB}$ from $\pi_x/\phihat$
  to $a^S(\pi_x)$, which is smaller as $a^S(z) < z/\phihat$ for
  $z > 0$.
\end{proof}

{\singlespacing\fontsize{12}{13.8}\selectfont
\bibliographystyle{chicago}
\bibliography{project}}

@article{holmstrom1987aggregation,
  title={Aggregation and linearity in the provision of intertemporal incentives},
  author={Holmstr\"{o}m, Bengt and Milgrom, Paul},
  journal={Econometrica},
  volume={55},
  number={2},
  pages={303--328},
  year={1987},
  publisher={JSTOR}
}

@article{williams2015solvable,
  title={A solvable continuous time dynamic principal--agent model},
  author={Williams, Noah},
  journal={Journal of Economic Theory},
  volume={159},
  pages={989--1015},
  year={2015},
  publisher={Elsevier}
}

@article{myerson2015moral,
  title={Moral hazard in high office and the dynamics of aristocracy},
  author={Myerson, Roger B},
  journal={Econometrica},
  volume={83},
  number={6},
  pages={2083--2126},
  year={2015},
  publisher={Wiley Online Library}
}

@article{ditella2021optimal,
  title   = {Optimal Asset Management Contracts with Hidden Savings},
  author  = {Di Tella, Sebastian and Sannikov, Yuliy},
  journal = {Econometrica},
  year    = {2021},
  volume  = {89},
  number  = {3},
  pages   = {1099--1139},
  month   = {5},
  doi     = {10.3982/ECTA14929}
}

@article{kocherlakota2004figuring,
  title={Figuring out the impact of hidden savings on optimal unemployment insurance},
  author={Kocherlakota, Narayana R},
  journal={Review of Economic Dynamics},
  volume={7},
  number={3},
  pages={541--554},
  year={2004},
  publisher={Elsevier}
}

@article{bromberg2021scale,
  title={Scale effects in dynamic contracting},
  author={Bromberg-Silverstein, Shirley and Moreno-Bromberg, Santiago and Roger, Guillaume},
  journal={Mathematics and Financial Economics},
  volume={15},
  pages={431--472},
  year={2021},
  publisher={Springer}
}

@book{brooks1975mythical,
  author    = {Brooks, Frederick P., Jr.},
  title     = {The Mythical Man-Month: Essays on Software Engineering},
  year      = {1975},
  publisher = {Addison-Wesley},
  address   = {Reading, MA}
}

@book{cusumano1991japan,
  author    = {Cusumano, Michael A.},
  title     = {Japan's Software Factories: A Challenge to {U.S.} Management},
  year      = {1991},
  publisher = {Oxford University Press},
  address   = {New York}
}

@article{toxvaerd2006time,
  title={Time of the essence},
  author={Toxvaerd, Flavio},
  journal={Journal of Economic Theory},
  volume={129},
  number={1},
  pages={252--272},
  year={2006},
  publisher={Elsevier}
}

@article{rochet2016risky,
  title={Risky utilities},
  author={Rochet, Jean-Charles and Roger, Guillaume},
  journal={Economic Theory},
  volume={62},
  pages={361--382},
  year={2016},
  publisher={Springer}
}

@article{georgiadis2014project,
  title={Project design with limited commitment and teams},
  author={Georgiadis, George and Lippman, Steven A and Tang, Christopher S},
  journal={The RAND Journal of Economics},
  volume={45},
  number={3},
  pages={598--623},
  year={2014},
  publisher={Wiley Online Library}
}

@article{li2025optimal,
  title={Optimal contracts with hidden risk},
  author={Li, Rui and Williams, Noah},
  journal={Review of Economic Dynamics},
  volume={58},
  pages={101306},
  year={2025},
  publisher={Elsevier}
}

@article{biais2010large,
  title={Large risks, limited liability, and dynamic moral hazard},
  author={Biais, Bruno and Mariotti, Thomas and Rochet, Jean-Charles and Villeneuve, St{\'e}phane},
  journal={Econometrica},
  volume={78},
  number={1},
  pages={73--118},
  year={2010},
  publisher={Wiley Online Library}
}

@book{jeanblanc2009mathematical,
  title={Mathematical methods for financial markets},
  author={Jeanblanc, Monique and Yor, Marc and Chesney, Marc},
  year={2009},
  publisher={Springer Science \& Business Media}
}

@article{feng2024setbacks,
  title={Setbacks, shutdowns, and overruns},
  author={Feng, Felix Zhiyu and Taylor, Curtis R and Westerfield, Mark M and Zhang, Feifan},
  journal={Econometrica},
  volume={92},
  number={3},
  pages={815--847},
  year={2024},
  publisher={Wiley Online Library}
}

@article{wong2019dynamic,
  title={Dynamic agency and endogenous risk-taking},
  author={Wong, Tak-Yuen},
  journal={Management Science},
  volume={65},
  number={9},
  pages={4032--4048},
  year={2019},
  publisher={INFORMS}
}

@article{marinovic2019ceo,
  title={{CEO} horizon, optimal pay duration, and the escalation of short-termism},
  author={Marinovic, Iv{\'a}n and Varas, Felipe},
  journal={The Journal of Finance},
  volume={74},
  number={4},
  pages={2011--2053},
  year={2019},
  publisher={Wiley Online Library}
}

@article{cetemen2023renegotiation,
  title={Renegotiation and dynamic inconsistency: contracting with non-exponential discounting},
  author={Cetemen, Doruk and Feng, Felix Zhiyu and Urgun, Can},
  journal={Journal of Economic Theory},
  volume={208},
  pages={105606},
  year={2023},
  publisher={Elsevier}
}

@article{he2017optimal,
  title={Optimal long-term contracting with learning},
  author={He, Zhiguo and Wei, Bin and Yu, Jianfeng and Gao, Feng},
  journal={The Review of Financial Studies},
  volume={30},
  number={6},
  pages={2006--2065},
  year={2017},
  publisher={Oxford University Press}
}

@article{brsvb12,
  title={Optimal search for product information},
  author={Branco, Fernando and Sun, Monic and Villas-Boas, J Miguel},
  journal={Management Science},
  volume={58},
  number={11},
  pages={2037--2056},
  year={2012},
  publisher={INFORMS}
}

@article{fudenberg1990short,
  title={Short-term contracts and long-term agency relationships},
  author={Fudenberg, Drew and Holmstr\"{o}m, Bengt and Milgrom, Paul},
  journal={Journal of Economic Theory},
  volume={51},
  number={1},
  pages={1--31},
  year={1990},
  publisher={Elsevier}
}

@techreport{demarzo2013risking,
  title={Risking other people's money: Gambling, limited liability, and optimal incentives},
  author={DeMarzo, Peter M and Livdan, Dmitry and Tchistyi, Alexei},
  institution={Stanford GSB},
  year={2014},
  number=3149
}

@techreport{bloedel2023persistent,
  title={Persistent Private Information Revisited},
  author={Bloedel, Alexander and Krishna, R. Vijay and Strulovici, Bruno},
  institution={UCLA, FSU, and Northwestern University},
  year={2023}
}

@article{mcclellan2022experimentation,
  title={Experimentation and approval mechanisms},
  author={McClellan, Andrew},
  journal={Econometrica},
  year    = {2022},
  volume  = {90},
  number  = {5},
  pages   = {2215--2247},
  doi     = {10.3982/ECTA17021}

}

@article{mason2015getting,
  title={Getting it done: dynamic incentives to complete a project},
  author={Mason, Robin and V{\"a}lim{\"a}ki, Juuso},
  journal={Journal of the European Economic Association},
  volume={13},
  number={1},
  pages={62--97},
  year={2015},
  publisher={Oxford University Press}
}

@article{halac2016optimal,
  title={Optimal contracts for experimentation},
  author={Halac, Marina and Kartik, Navin and Liu, Qingmin},
  journal={The Review of Economic Studies},
  volume={83},
  number={3},
  pages={1040--1091},
  year={2016},
  publisher={Wiley-Blackwell}
}

@article{bergemann1998venture,
  title={Venture capital financing, moral hazard, and learning},
  author={Bergemann, Dirk and Hege, Ulrich},
  journal={Journal of Banking \& Finance},
  volume={22},
  number={6-8},
  pages={703--735},
  year={1998},
  publisher={Elsevier}
}

@article{hopenhayn1997optimal,
  title={Optimal Unemployment Insurance},
  author={Hopenhayn, Hugo A and Nicolini, Juan Pablo},
  journal={Journal of Political Economy},
  volume={105},
  number={2},
  pages={412--438},
  year={1997}
}

@article{shavell1979optimal,
  title={The Optimal Payment of Unemployment Insurance Benefits},
  author={Shavell, Steven and Weiss, Laurence},
  journal={Journal of Political Economy},
  volume={87},
  number={6},
  pages={1347--1362},
  year={1979}
}

@article{sannikov2008continuous,
  title={A continuous-time version of the principal-agent problem},
  author={Sannikov, Yuliy},
  journal={The Review of Economic Studies},
  volume={75},
  number={3},
  pages={957--984},
  year={2008},
  publisher={Wiley-Blackwell}
}

@article{madsen2022designing,
  title={Designing deadlines},
  author={Madsen, Erik},
  journal={American Economic Review},
  volume={112},
  number={3},
  pages={963--997},
  year={2022},
  publisher={American Economic Association 2014 Broadway, Suite 305, Nashville, TN 37203}
}

@article{he2011model,
  title={A model of dynamic compensation and capital structure},
  author={He, Zhiguo},
  journal={Journal of Financial Economics},
  volume={100},
  number={2},
  pages={351--366},
  year={2011},
  publisher={Elsevier}
}

@article{georgiadis2015projects,
  title={Projects and team dynamics},
  author={Georgiadis, George},
  journal={The Review of Economic Studies},
  volume={82},
  number={1},
  pages={187--218},
  year={2015},
  publisher={Oxford University Press}
}

@article{manso2011motivating,
  title={Motivating Innovation},
  author={Manso, Gustavo},
  journal={Journal of Finance},
  volume={66},
  pages={1823--1860},
  year={2011}
}

@article{bergemann2005financing,
  title={The financing of innovation: Learning and stopping},
  author={Bergemann, Dirk and Hege, Ulrich},
  journal={RAND Journal of Economics},
  volume={36},
  number={4},
  pages={719--752},
  year={2005},
  publisher={JSTOR}
}

@article{horner2013incentives,
  title={Incentives for experimenting agents},
  author={H{\"o}rner, Johannes and Samuelson, Larry},
  journal={The RAND Journal of Economics},
  volume={44},
  number={4},
  pages={632--663},
  year={2013},
  publisher={Wiley Online Library}
}

@article{guo2016dynamic,
  title={Dynamic delegation of experimentation},
  author={Guo, Yingni},
  journal={American Economic Review},
  volume={106},
  number={8},
  pages={1969--2008},
  year={2016}
}

@article{green2016breakthroughs,
  title={Breakthroughs, deadlines, and self-reported progress: Contracting for multistage projects},
  author={Green, Brett and Taylor, Curtis R},
  journal={American Economic Review},
  volume={106},
  number={12},
  pages={3660--99},
  year={2016}
}

@Book{teschl2012ordinary,
	author = {Teschl, Gerald},
	publisher = {American Mathematical Society Providence},
	title = {Ordinary Differential Equations and Dynamical Systems},
	volume = 140,
	year = 2012
}

@Article{demarzo2011learning,
	author = {DeMarzo, Peter and Sannikov, Yuliy},
	journal = {The Review of Economic Studies},
	title = {Learning, termination and payout policy in dynamic incentive contracts},
	volume = {84},
	number = {1},
	pages = {182--236},
	year = 2017
}

@techreport{thomas2021clinical,
  author      = {Thomas, David and Chancellor, Daniel and Micklus, Amanda
                 and LaFever, Sara and Hay, Michael and Chaudhuri, Shomesh
                 and Bowden, Robert and Lo, Andrew W.},
  title       = {Clinical Development Success Rates and Contributing Factors 2011--2020},
  institution = {Biotechnology Innovation Organization (BIO), QLS Advisors,
                 and Informa Pharma Intelligence},
  year        = {2021},
  month       = feb,
  url         = {https://go.bio.org/rs/490-EHZ-999/images/ClinicalDevelopmentSuccessRates2011_2020.pdf}
}

@article{feng2025setbacks,
  title={Setbacks Big and Small},
  author={Feng, Felix Zhiyu and Taylor, Curtis R and Westerfield, Mark M and Zhang, Feifan},
  journal={Available at SSRN 5880422},
  year={2025}
}

@article{hellwig2002discrete,
  title={Discrete-time approximations of the {H}olmstr\"{o}m-{M}ilgrom {B}rownian-motion model of intertemporal incentive provision},
  author={Hellwig, Martin F. and Schmidt, Klaus M.},
  journal={Econometrica},
  volume={70},
  number={6},
  pages={2225--2264},
  year={2002}
}

\end{document}